\documentclass[acmsmall, authorversion]{acmart}

\AtBeginDocument{%
  }

\setcopyright{rightsretained}
\acmJournal{PACMMOD}
\acmYear{2024} \acmVolume{2} \acmNumber{N6 (SIGMOD)}
\acmArticle{235} \acmMonth{12} \acmPrice{15.00}
\acmDOI{10.1145/3698810}

\acmISBN{978-1-4503-XXXX-X/18/06}

\usepackage{algorithm}
\usepackage{algpseudocode}
\usepackage{amsthm}
\usepackage{array}
\usepackage{balance}
\usepackage{booktabs}
\usepackage{color}
\usepackage{caption}
\usepackage{datatool}
\usepackage{fontawesome}
\usepackage{tikz}
\usepackage{graphics}
\usepackage{pgfplots}\pgfplotsset{compat=1.18}
\usepackage{fancybox}
\usepackage{listings}
\usepackage{xspace}
\usepackage{multirow}
\usepackage{xcolor}
\usepackage{tcolorbox}
\usepackage{subfigure}
\usepackage{graphicx}
\usetikzlibrary{shapes.geometric, arrows.meta, positioning, fit, calc, backgrounds, patterns}
\tcbuselibrary{listings}

\DTLloaddb{status}{data/status.csv}
\DTLloaddb{type}{data/type.csv}

\begin{document}
\begin{filecontents*}{data/status.csv}

       status, Fixed, Confirmed, Unconfirmed, Duplicate, Sum 
      GEOS,    4, 8, 0, 0, 12
      PostGIS, 8, 1, 1, 1, 11
DuckDBSpatial, 5, 0, 1, 0, 6
     MySQLGIS, 1, 3, 0, 0, 4
    SQLSERVER, 0, 0, 2, 0, 2
          Sum, 18, 12, 4, 1, 35
   Identified, 0, 0, 0, 0, 34
         Real, 0, 0, 0, 0, 30
\end{filecontents*}

\begin{filecontents*}{data/type.csv}

         type, LogicFixed, LogicConfirmed, CrashFixed, CrashConfirmed, Sum
      GEOS,    1, 8, 3, 0, 12
      PostGIS, 6, 1, 2, 0, 9
DuckDBSpatial, 0, 0, 5, 0, 5
     MySQLGIS, 1, 3, 0, 0, 4
          Sum, 8, 12, 10, 0, 30
        Logic, 0, 0, 0, 0, 20
        Crash, 0, 0, 0, 0, 10
\end{filecontents*}

\newcommand\DBMS{{\textsc{DBMS}}}
\newcommand\SDBMS{{\textsc{SDBMS}}}
\newcommand\PostGIS{{{PostGIS}}}
\newcommand\PostgreSQL{{{PostgreSQL}}}
\newcommand\DuckDB{{{DuckDB}}}
\newcommand\DuckDBSpatial{{{DuckDB Spatial}}}
\newcommand\MySQLGIS{{{MySQL}}}
\newcommand\MySQL{{{MySQL}}}
\newcommand\SQLSERVER{{{SQL Server}}}
\newcommand\JTS{{\textsc{JTS}}}
\newcommand\GEOS{{\textsc{GEOS}}}
\newcommand\ToolName{{{Spatter}}}
\newcommand\OGC{\emph{Open Geospatial Consortium Standards}}
\newcommand\OGCAKA{{OGC}}
\newcommand\DENIM{{\textsc{DE-9IM}}}
\newcommand\APPROACHNAME{\emph{Affine Equivalent Inputs}}
\newcommand\APPROACHAKA{\emph{AEI}}
\newcommand\SMARTGENERATORUPPER{{Geometry-Aware Generator}}
\newcommand\SMARTGENERATOR{\emph{geometry-aware generator}}
\newcommand\SMARTGENERATORAKA{\emph{GAG}}
\newcommand\RANDOMGENERATORAKA{\emph{RSG}}
\newcommand\genbasedstrgy{\emph{random-shape strategy}}
\newcommand\derivativestrgy{\emph{derivative strategy}}
\newcommand\EmptyRemoval{\emph{EMPTY removal}}  \newcommand\CapEmptyRemoval{{EMPTY Removal}}
\newcommand\Homogenization{\emph{homogenization}}  \newcommand\CapHomogenization{{Homogenization}} 
\newcommand\DuplicateRemoval{\emph{duplicate removal}} \newcommand\CapDuplicateRemoval{{Duplicate Removal}}
\newcommand\Reordering{\emph{reordering}} \newcommand\CapReordering{{Reordering}}
\newcommand\ConsecutiveDuplicateRemoval{\emph{consecutive duplicate removal}}
\newcommand\CapConsecutiveDuplicateRemoval{Consecutive Duplicate Removal}

\newcommand{\sqlinline}[1]{\lstinline[ language=SQL,breaklines=true, basicstyle=\ttfamily\small\color{sqlinlineColor},keywordstyle={}]{#1}}

\newcommand{\EMPTYPATTERN}{6}
\newcommand{\MIXEDPATTERN}{13}
\newcommand{\OVERLOOKED}{14}
\newcommand{\SDB}{SDB}

\newcommand{\ignorelst}[1]{#1}

\algtext*{EndIf} 
\algtext*{EndWhile} 
\algtext*{EndFor} 
\algtext*{EndFunction}

\definecolor{NavyBlue}{rgb}{0,0,0.5}
\definecolor{lstString}{HTML}{576574}
\definecolor{myGreen}{HTML}{009432}
\definecolor{bugColor}{HTML}{c0392b}
\definecolor{sqlinlineColor}{HTML}{000000}

\newcommand{\bugPattern}[1]{\textcolor{bugColor}{\emph{#1}}}

\renewcommand*\Call[2]{\textproc{\textnormal{#1}}(#2)}

\lstdefinestyle{sqlstyle}{
    language=SQL,
	basicstyle=\small\ttfamily,	
    keywordstyle=\color{NavyBlue},
    stringstyle=\small\color{lstString},
    backgroundcolor=\color{white!5},
	frame=single,
	framerule=0pt,
    rulecolor=\color{black},
	numbers=left,
    numberstyle=\footnotesize,
    firstnumber=1,
    stepnumber=1,
    numbersep=7pt,
	showspaces=false,
	showstringspaces=false,
	keepspaces=true,		
	showtabs=false,
	tabsize=4,
	flexiblecolumns=true,
	breaklines=true,
	breakatwhitespace=false,
	breakautoindent=true,
	breakindent=1em,
	title=\lstname,	
	escapeinside=``,
	xleftmargin=2em,  
    xrightmargin=0em,
	aboveskip=0ex, 
    belowskip=0ex,
    framexleftmargin=4ex,
	framextopmargin=0pt,
    framexbottommargin=0pt,
    abovecaptionskip=0pt,
    belowcaptionskip=2pt,
	extendedchars=false, 
    columns=flexible, mathescape=false,
	texcl=false,
	fontadjust,
    captionpos=t,
}

\captionsetup[lstlisting]{font=small, singlelinecheck=false, justification=raggedright}

%%
%% The "title" command has an optional parameter,
%% allowing the author to define a "short title" to be used in page headers.
\title[Finding Logic Bugs in Spatial Database Engines]{Finding Logic Bugs in Spatial Database Engines \\
via \APPROACHNAME{}}
%%
%% The "author" command and its associated commands are used to define
%% the authors and their affiliations.
%% Of note is the shared affiliation of the first two authors, and the
%% "authornote" and "authornotemark" commands
%% used to denote shared contribution to the research.
\author{Wenjing Deng}
\orcid{0009-0004-0433-6860}
\affiliation{%
  \institution{East China Normal University}
  \city{Shanghai}
  \country{China}}
\email{51215902117@stu.ecnu.edu.cn}

\author{Qiuyang Mang}
\orcid{0009-0000-4462-8527}
\affiliation{%
  \institution{The Chinese University of Hong Kong, Shenzhen}
  \city{Guangdong}
  \country{China}}
\email{qiuyangmang@link.cuhk.edu.cn}

\author{Chengyu Zhang}
\orcid{0000-0002-7285-289X}
\affiliation{%
  \institution{ETH Zurich}
  \city{Zurich}
  \country{Switzerland}
}
\email{chengyu.zhang@inf.ethz.ch}

\author{Manuel Rigger}
\orcid{0000-0001-8303-2099}
\affiliation{%
 \institution{National University of Singapore}
 \city{Singapore}
 \country{Singapore}
}
\email{rigger@nus.edu.sg}

\begin{abstract}
Spatial Database Management Systems (\SDBMS{}s) aim to store, manipulate, and retrieve \emph{spatial data}. \SDBMS{}s are employed in various modern applications, such as geographic information systems, computer-aided design tools, and location-based services. However, the presence of logic bugs in \SDBMS{}s can lead to incorrect results, substantially undermining the reliability of these applications. Detecting logic bugs in \SDBMS{}s is challenging due to the lack of ground truth for identifying incorrect results. In this paper, we propose an automated \SMARTGENERATOR{} to generate high-quality SQL statements for \SDBMS{}s and a novel concept named \APPROACHNAME{} (\APPROACHAKA{}) to validate the results of \SDBMS{}s. We implemented them as a tool named \ToolName{} (\underline{Spat}ial \DBMS{}s Tes\underline{ter}) for finding logic bugs in four popular \SDBMS{}s: \PostGIS{}, \DuckDBSpatial{}, \MySQLGIS{}, and \SQLSERVER{}. Our testing campaign detected \DTLfetch{status}{status}{Identified}{Sum} previously unknown and unique bugs in these \SDBMS{}, of which \DTLfetch{status}{status}{Real}{Sum} have been confirmed, and \DTLfetch{status}{status}{Sum}{Fixed} have been already fixed. Our testing efforts have been well appreciated by the developers. Experimental results demonstrate that the \SMARTGENERATOR{} significantly outperforms a naive random-shape generator in detecting unique bugs, and \APPROACHAKA{} can identify \OVERLOOKED{} logic bugs in \SDBMS{}s that were overlooked by previous methodologies.
\end{abstract}

%%
%% The code below is generated by the tool at http://dl.acm.org/ccs.cfm.
%% Please copy and paste the code instead of the example below.
%%


\received{18 April 2024}
\received[revised]{25 July 2024}
\received[accepted]{1 August 2024}

\maketitle

\section{Introduction}
Spatial Database Management Systems (\SDBMS{}s) aim to store, manipulate, and retrieve \emph{spatial data}, which describes objects and locations under a coordinate system~\cite{guting1994introduction}. \SDBMS{}s are employed in various modern applications, such as geographic information systems~\cite{rigaux2002spatial}, computer-aided design~\cite{10.1145/376284.375777}, location-based services~\cite{Zheng2010GeoLifeAC}, and scientific simulations~\cite{10.1145/2463676.2463700, 10.1145/2723372.2749434}.
\SDBMS{}s are typically implemented as spatial extensions or build-in features of widely-used relational \DBMS{}s.
For example, \PostGIS{}, the most popular \SDBMS{} on DB-Engines Ranking, is a spatial extension of \PostgreSQL{}; \MySQL{}, one of the most widely-used relational \DBMS{}s supports spatial built-in features.

Despite the importance of \SDBMS{}s, their reliability has not received sufficient attention.
While general-purpose fuzzers have been applied to \SDBMS{}s for generating obscure inputs that cause crashes~\cite{postgis-oss-fuzz2017}---namely crash bug detection---they have failed to detect logic bugs.
Logic bugs are a particularly notorious category of bugs, causing the \SDBMS{}s to compute incorrect results.
Unlike crash bugs, which terminate the process,
logic bugs silently produce incorrect results and thus often go unnoticed by both developers and application users.

Listing~\ref{lst:precision-issue} shows statements that trigger a logic bug in \PostGIS{}. 
The statements insert a line segment and a point into the database (Line~\ref{line:create-start}--\ref{line:create-end}) and query whether the line segment covers the point (Line~\ref{line:cover-query}).
The line segment and the point are drawn in Figure~\ref{fig:intro_subfig1}, clearly showing that the line covers the point, which means the result should be \sqlinline{1}. 
However, \PostGIS{} incorrectly returned \sqlinline{0}, which indicates a logic bug.

\begin{figure}[t]
\lstset{style=sqlstyle}
\begin{lstlisting}[caption={Statements that trigger a bug in \PostGIS{}. The retrieved value from \PostGIS{} should be 1 instead of 0.}, label={lst:precision-issue}]
CREATE TABLE t1 (g geometry); `\label{line:create-start}`
CREATE TABLE t2 (g geometry);
INSERT INTO t1 (g) VALUES ('LINESTRING(0 1,2 0)');
INSERT INTO t2 (g) VALUES ('POINT(0.2 0.9)'); `\label{line:create-end}`
SELECT COUNT(*) FROM t1 JOIN t2 ON ST_Covers(t1.g,t2.g);`\label{line:cover-query}`
-- {0} `\faBug{}`  {1} `\faCheckCircle{}`
\end{lstlisting}
\end{figure}

\begin{figure}[t]
\lstset{style=sqlstyle}
\begin{lstlisting}[caption={The statements correspond to the geometries in Figure~\ref{fig:intro_subfig2}. }, label={lst:AEI_example}]
CREATE TABLE t1 (g geometry); 
CREATE TABLE t2 (g geometry);
INSERT INTO t1 (g) VALUES ('LINESTRING(1 1,0 0)');
INSERT INTO t2 (g) VALUES ('POINT(0.9 0.9)');
SELECT COUNT(*) FROM t1 JOIN t2 ON ST_Covers(t1.g,t2.g);
-- {1} `\faCheckCircle{}`
\end{lstlisting}
\end{figure}

\begin{figure}[t]
\centering
\subfigure[The line covers the point.]{
\includegraphics[width=0.30\textwidth]{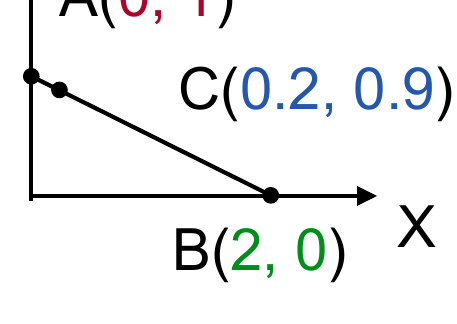}
\label{fig:intro_subfig1}
}
\hspace{15pt}
\subfigure[Geometries affine transformed from left geometries.]{
\includegraphics[width=0.30\textwidth]{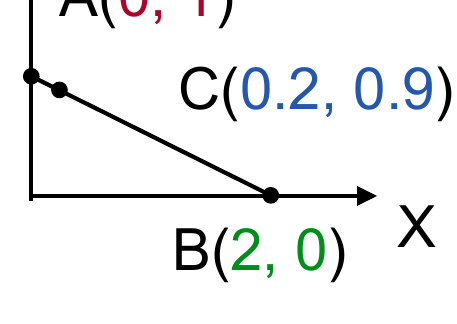}
\label{fig:intro_subfig2}
}
\captionsetup{font={small, bf}}
\caption{Visualizations of the geometries in Listing~\ref{lst:precision-issue} and Listing~\ref{lst:AEI_example} are shown in (a) and (b) respectively. The geometry pairs in (a) and (b) are affine equivalent.}
\label{fig:intro_figures}
\end{figure}

Detecting logic bugs automatically in \SDBMS{}s remains a challenging problem.
To the best of our knowledge, 
in \SDBMS{}s,
past practices largely rely on user reports to identify logic bugs; no automated testing strategies have been applied.  
The key challenge of automatically finding logic bugs in \SDBMS{}s is the lack of ground truth results.
Many automated testing techniques have been proposed for detecting logic bugs in relational \DBMS{}s~\cite{10.1145/3551349.3556924, 10.14778/3357377.3357382, 10.1109/ICSE48619.2023.00174, 10.1145/3368089.3409710, 10.1145/3428279, 10.5555/3488766.3488804, 10.1145/3510003.3510093, 10.1109/ICSE48619.2023.00175},
but unfortunately, they are not applicable for testing \SDBMS{}s, especially for spatial-related features.

One methodology of testing relational \DBMS{} works is to generate the query, pass it to different systems, and consider the equivalence of their outputs as the expected result, a technique known as \emph{differential testing}~\cite{10.1145/3551349.3556924, 10.14778/3357377.3357382}. 
However, for spatial-related features that are solely implemented in one \SDBMS{}, 
differential testing techniques are inapplicable, because the expected result cannot be constructed.
In addition, bugs can also be in the shared third-party libraries, 
leading multiple \SDBMS{}s to produce consistent, but incorrect outputs, 
resulting in bugs being missed. 
Furthermore, although the \OGC{}~(\OGCAKA{}) have standardized most functions and data types in \SDBMS{}s, implementation details still vary among \SDBMS{}s~\cite{piorkowski2011mysql}, which means their outputs are expected to be nonequivalent, resulting in \emph{false alarms}. 
We observed various expected discrepancies between \SDBMS{}s---as per developers' intents.
The bug in Listing~\ref{lst:precision-issue} cannot be detected by differential testing, 
as function \sqlinline{ST_Covers} is only implemented in \PostGIS{} and \DuckDBSpatial{}, relying on their common third-party library. 

Another methodology involves generating statements and constructing their expected results within the same \DBMS{}, thus avoiding the limitations of differential testing. 
For instance, Ternary Logic Partitioning (TLP) is a general state-of-the-art testing technique for relational \DBMS{}s~\cite{10.1145/3428279}, applicable not only to relational \DBMS{}s, but also to graph \DBMS{}s~\cite{10.1145/3597926.3598044}. 
The key insight of TLP is to derive three partitioning queries from an original query, and the sum of the results from the partitioned queries is expected to be the same as that of the original one. 
However, TLP may fail to detect logic bugs in spatial-related features. 
For example, the bug presented in Listing~\ref{lst:precision-issue} cannot be found by TLP because both the summed-up results of the partitioning queries and the original query result are incorrect. 
Therefore, an \SDBMS-oriented approach for generating test cases and expected results is necessary.

We propose a methodology named \APPROACHNAME{} (\APPROACHAKA{}) to provide the expected results for \SDBMS{}s. Our key insight is that if two geometries affine transform (\emph{e.g.}, rotate, scale, and translate) in the same way, \emph{topological relationships} (\emph{e.g.}, intersects, covers, or disjoint) are preserved.  
Therefore, two topological relationships of two \emph{affine equivalent} geometry pairs retrieved from an \SDBMS{} are expected to be equal; otherwise, a bug is detected. 
We consider two geometry pairs to be affine equivalent if they can be transformed into each other through an affine transformation.

For example, considering the geometries in Figure~\ref{fig:intro_subfig1}, affine transforming them yields the geometries in Figure~\ref{fig:intro_subfig2} corresponding to the statements in Listing~\ref{lst:AEI_example}.
Since affine transformations are invertible~\cite{affine-transformation}, the geometry pairs in Figure~\ref{fig:intro_subfig1} and Figure~\ref{fig:intro_subfig2} are deemed \emph{affine equivalent}. 
With the expectation of equal topological relationships, the statements in \text{Listing \ref{lst:precision-issue} and \ref{lst:AEI_example}} should yield the same results. 
However, \PostGIS{} correctly evaluates \sqlinline{1} for the statements in Listing~\ref{lst:AEI_example} and incorrectly returns \sqlinline{0} for the statements in Listing~\ref{lst:precision-issue}, revealing a logic bug. 
We found this bug, which was caused by precision issues, via \APPROACHAKA{}. 
Specifically, the bug was caused by a loss of precision in the normalization of vertices (\emph{i.e.}, displacement to the origin). 
The statements in Listing~\ref{lst:AEI_example} fail to trigger this bug, because a point in $v$ is at the origin, leading to no displacement calculations. 
Overall, in this example, \APPROACHAKA{} reveals the bug by triggering different execution paths (\emph{i.e.}, with/without displacement calculations). 
As demonstrated through the above example, such precision issues may result in inaccuracies in topological relationship calculations.
The developers of \PostGIS{} have acknowledged this bug and have highlighted their concern by modifying the issue title to \emph{"Covers predicate fails on obviously correct simple case"}.
% \textbf{obviously correct simple case}"
After we reported this bug, the developers indicated that they would start working on a new predicate evaluation mechanism, which would include a tolerance argument to prevent such issues.

Based on \APPROACHAKA{}, we implemented a tool named \ToolName{} for testing \SDBMS{}s.\footnote{Our artifact is publicly available at \url{https://github.com/cuteDen-ECNU/Spatter} and \url{https://zenodo.org/records/13932460}.}
\ToolName{} consists of a \SMARTGENERATOR{} for generating high-quality SQL statements for \SDBMS{}s and uses \APPROACHAKA{} to validate the results of \SDBMS{}s.  
To the best of our knowledge, \ToolName{} is the first automated testing tool to detect logic bugs in \SDBMS{}s.
To evaluate the effectiveness of \APPROACHAKA{}, we selected 4 widely-used \SDBMS{}s as our testing targets: \PostGIS{}, \MySQLGIS{}, \DuckDBSpatial{}, and \SQLSERVER{}. 
The results show that our approach is effective in detecting spatial-related logic bugs in \SDBMS{}.
Our testing campaign detected \DTLfetch{status}{status}{Identified}{Sum} real, previously unknown, and unique bugs, which were missed by existing test suites.
\DTLfetch{status}{status}{Real}{Sum} of them have been confirmed, 
\DTLfetch{status}{status}{Sum}{Fixed} have been fixed, 
and \DTLfetch{type}{type}{Logic}{Sum} were logic bugs. 
During the testing campaign,
our work received positive feedback from the \SDBMS{}s’ developers.
Our experimental results demonstrate that the \SMARTGENERATOR{} significantly outperforms a naive random-shape generator
in detecting unique bugs, 
and \APPROACHAKA{} could
identify \OVERLOOKED{} logic bugs in \SDBMS{}s that were totally overlooked by
previous methodologies.

To summarize, we make the following contributions:
\begin{itemize}
    \item We propose \APPROACHNAME{}---a methodology to automatically validate the results of \SDBMS{}s.
    \item We designed and implemented an automated tool \ToolName{} for detecting logic bugs in \SDBMS{}s.
    \item We detected \DTLfetch{status}{status}{Identified}{Sum} previously-unknown and unique bugs in four widely-used \SDBMS{}s.
    \item We evaluated \ToolName{} and show that it outperforms the baseline strategies. 
\end{itemize}

\section{Background}
In this section, we provide important background information on geometry types, topological relationship queries, and affine transformations.

\subsection{Geometry Types\label{sec:geometry-type}}

\SDBMS{}s provide geometry types, on which we focus, given that they are widely used in real-world datasets~\cite{10.1093/nar/gkz934} and support various functions.
Figure~\ref{fig:geometry-type} illustrates seven widely-used 2D geometry types; these and other geometry types are standardized by \OGC{} (\OGCAKA{})~\cite{OGC2006WMS}.
The basic geometry types shown on the left include {\small POINT}, {\small LINESTRING}, and {\small POLYGON}, with dimensions of 0, 1, and 2 respectively. 
A geometry whose type starts with {"MULTI"} consists of multiple elements of the same basic type; we refer to this as \emph{MULTI geometry}.
The type {\small GEOMETRYCOLLECTION} comprises elements of mixed geometry type; we refer to this as \emph{MIXED geometry}.
Examples of \emph{MULTI} and \emph{MIXED} geometries are shown on the right side of Figure~\ref{fig:geometry-type}.

\begin{figure}[t]
\centering
\includegraphics[width=.65\textwidth]{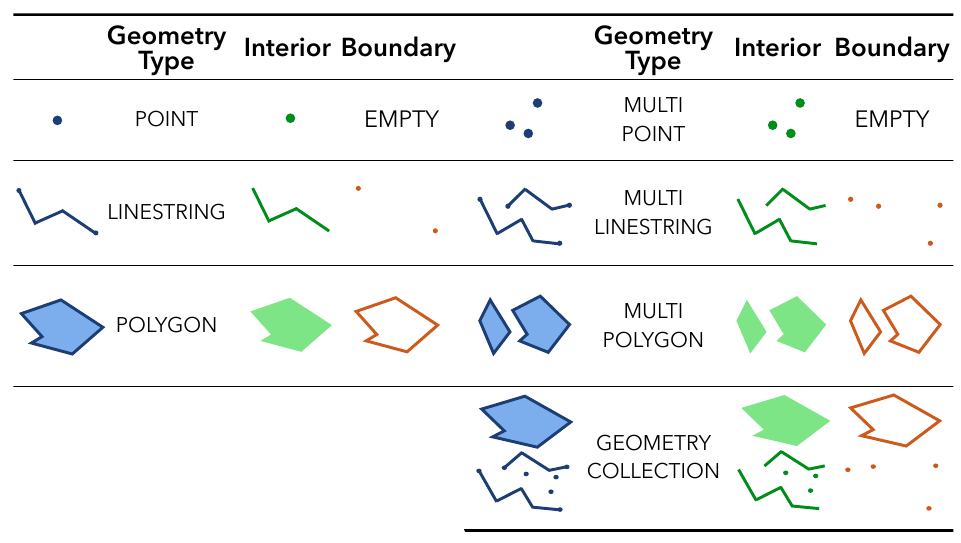}
\captionsetup{font={small, bf}}
\caption{Examples of 2D geometries and their geometry types, interiors, and boundaries. \label{fig:geometry-type} }
\end{figure}

\subsection{Topological Relationship Queries\label{sec:topo-rlt}}
Topological relationships represent qualitative properties that characterize the relative position of spatial objects~\cite{Felice2009}.
The concept of a \emph{topological relationship query} is a core feature when operating on spatial data, 
playing a key role in spatial queries and joins. 
Therefore, the correctness of the execution engine in a \SDBMS{} that processes such topological relationship queries is crucial. 
Despite this, to the best of our knowledge, 
no automated testing approach for detecting logic bugs in such functionality has been proposed.
\emph{Formal topological relationships} and \emph{named topological relationships} are two types of widely-used topological relationships in \SDBMS{}s. 

\paragraph{\textbf{Definition of formal topological relationships}}
\emph{Formal topological relationships} are formally defined in the Dimensionally Extended 9-Intersection Model (\DENIM{}~\cite{10.1007/3-540-56869-7_16, CLEMENTINI1994815}),
which establishes spatial relationships between two given geometries. 

In \DENIM{}, the geometry is conceptualized as a \emph{cell complex} using principles from algebraic topology~\cite{9-intersection-1990}.
A cell complex consists of cells and their respective \emph{faces}. The faces of a cell are crucial components that delineate the cell's boundaries.~\cite{10.1007/3-540-52208-5_32}.
Cells are primitive geometric objects defined across various spatial dimensions. Specifically, an $n$-cell is composed of $n+1$ geometrically independent cells of dimension $n-1$~\cite{10.1007/3-540-52208-5_32}.
A \emph{face} of $n$-cell $C$ is any $r$-cell contained within $C$, where $0 \leqslant r \leqslant n$~\cite{10.1007/3-540-52208-5_32}. 

Examples of cells are 0-cells representing points, 1-cells for edges, 2-cells for triangles, and 3-cells for tetrahedrons.
For example, a 2-cell has \emph{faces} that include two \emph{1-cells} (edges) and one \emph{2-cell}, representing its end-points and the edge connecting these end-points respectively.
In \SDBMS{} implementations, a \sqlinline{POINT} is a 0-cell; a \sqlinline{LINESTRING} is described as a cell complex composed of multiple 1-cells and their faces; a \sqlinline{POLYGON} is represented as a cell complex consisting of several 2-cells and their faces.

\emph{Formal topological relationships} are qualitatively defined based on the concepts of cell complexes' boundary, interior, and exterior.
Firstly, we precisely present the definition of \emph{boundary}, \emph{interior}, and \emph{exterior} for cells in Definition~\ref{def:cells}.
For instance, in a triangle (\emph{i.e.}, a 2-cell), the boundary consists of its three vertices and three edges; 
the interior is the area it encloses; 
the exterior is defined as the set from which the boundary and interior are excluded.
On top of Definition~\ref{def:cells}, Definition~\ref{def:cell-complexes} further elucidates the concepts of boundary, interior, and exterior within cell complexes. 

\begin{definition}[Boundary, Interior, and Exterior of Cells]
The closure of an $n$-cell $C$, denoted by $\overline{C}$, is the set of all \emph{r-faces} $f(r)$ of $C$, where $0 \leqslant r \leqslant n$, that is, $\overline{C} = \bigcup_{r = 0}^{n} f(r)$. 
The boundary of $C$, denoted by $\partial C$, is a finite union set of r-faces of $C$, where $0 \leqslant r \leqslant (n-1)$, that is, $\partial C = \bigcup_{r = 0}^{n - 1} f(r)$.
The interior of a cell $C$, denoted by $C^o$, is the difference set of closure and boundary of $C$, that is, $C^o = \overline{C} - \partial C$. 
The exterior of a cell $C$, denoted by $C^-$, is the difference set of universal set $\mathbb{U}$ and closure of $C$, that is, $C^- = \mathbb{U} - \overline{C}$~\cite{10.1007/3-540-52208-5_32}.
\label{def:cells}
\end{definition}

\begin{definition}[Boundary, Interior, Exterior of Cell Complexes]
    Let x be the number of cells $(C_1...C_x)$ that constitute a cell complex $G$.
    The boundary of $G$, denoted by $\partial G$, is the union of all the boundaries of $C_i$, while subtracting the intersections between boundaries of any distinct cell pairs, that is, 
    {
    % \small
    % \begin{equation}
        $\partial G = \bigcup_{i = 1}^x \partial C - \bigcup_{i = 1}^x\bigcup_{j = i+1}^x(\partial C_i \cap \partial C_j).$
    % \end{equation}
    }
    The interior of $G$, denoted as $G^o$, is the union set of the closure of $C$ (\emph{i.e.}, $\overline{C}$), subtracting $\partial G$, that is, $G^o = \bigcup_{i = 1}^x\overline{C_i} - \partial G$. 
    The exterior of $G$ is the intersection of $C_i^-$, that is, $G^- = \bigcap_{i=1}^x C_i^-$~\cite{10.1007/3-540-52208-5_32}. 
\label{def:cell-complexes}
\end{definition}

Figure~\ref{fig:geometry-type} intuitively demonstrates cell complexes, namely geometries as well as their interiors and boundaries.
Specifically, the boundary of {\small POINT} is defined as an empty set. 
The {\small LINESTRING} in the figure is a cell complex composed of three connected line segments (1-cells), with its boundary defined by its two endpoints. 
The \sqlinline{POLYGON} in the figure is a cell complex consisting of multiple triangles, with its boundary including its vertices and edges.

% \paragraph{\DENIM{}} 
\DENIM{} qualitatively defines and determines the \emph{topological relationship} between two spatial entities. %~\ref{def:cell-complexes}. 
It does so by characterizing the relationships between two cell complexes through the interactions among their boundaries (denoted by $g^o$), interiors (denoted by $\partial g$), and exteriors (denoted by $g^-$). 
The model employs a $3\times3$ matrix, which is constructed with the aid of a dimension calculator, denoted as $\mathcal{D}$.
When the intersection between the compared elements is empty, the calculator yields a constant value, F; otherwise, it returns $n$, a value that represents the dimension of the intersection.

\begin{definition}[Formal Topological Relationship]
For any two cell complexes $g_1$ and $g_2$, 
\begin{equation}
\resizebox{0.4\textwidth}{!}{
$
R(g_1, g_2) = 
\begin{bmatrix}
   \mathcal{D}[g_1^o \cap g_2^o] & \mathcal{D}[g_1^o \cap \partial g_2] &  \mathcal{D}[g_1^o \cap g_2^-] \\
   \mathcal{D}[\partial g_1 \cap g_2^o] & \mathcal{D}[\partial g_1 \cap \partial g_2] &  \mathcal{D}[\partial g_1 \cap g_2^-] \\
   \mathcal{D}[g_1^- \cap g_2^o] & \mathcal{D}[g_1^- \cap \partial g_2] &  \mathcal{D}[g_1^- \cap g_2^-]
\end{bmatrix}
$
},
\label{eq:de9im}
\end{equation}
where $\mathcal{D}$ is a dimension calculator.
\label{def:de9im}
\end{definition}

\begin{figure}[!t]
\centering
\includegraphics[width=.5\textwidth]{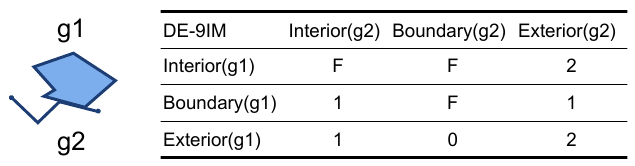}
\captionsetup{font={small, bf}}
\caption{ DE-9IM code of {\small POLYGON} g1 and {\small LINESTRING} g2 is \sqlinline{FF21F1102}. The letter \sqlinline{F} indicates that the intersection is an empty set. The numbers \sqlinline{0}, \sqlinline{1}, and \sqlinline{2} represent the dimension of the intersection.  \label{fig:de9im}} 
\end{figure}

Figure~\ref{fig:de9im} illustrates the \emph{formal topological relationship} between two example geometries $g1$ and $g2$. 
For instance, the dimension of the intersection between the boundary of g1 and the interior of g2 is 1 (\emph{i.e.}, $\partial g_1 \cap g_2^o = 1$), 
that is, the boundary of g1 and the interior of g2 intersect along an edge of dimension 1. Furthermore, 
the intersection of the interiors of g1 and g2 is an empty set (\emph{i.e.}, $g_1^o \cap g_2^o = F$), 
indicating that the interior of $g1$ does not intersect with the interior of $g2$.
In \PostGIS{}, the function \sqlinline{ST_Relate(g1,g2)}, which is implemented based on \DENIM{}, evaluates the formal topological relationship between $g1$ and $g2$. This relationship is represented as a string of length 9 with a domain of $\{0, 1, 2, F\}$. 
Specifically, the \DENIM{} code of g1 and g2 is \sqlinline{FF21F1102}.

\paragraph{\textbf{Named topological relationships}} \emph{Named topological relationships} are derived from \emph{formal topological relationships}.
Unlike \emph{formal topological relationships}, 
they describe relationships in a manner that is easily understood by users~\cite{10.1007/3-540-56869-7_16}.
A set of named topological relationships, such as functions \sqlinline{ST_Disjoint} and \sqlinline{ST_Intersects} are widely supported by \SDBMS{}s. 
Besides, \SDBMS{}s extend their functionality through specific named topological relationship queries. 
For example, as seen in Listing~\ref{lst:precision-issue}, the function \sqlinline{ST_Cover} is only implemented in \PostGIS{} and \DuckDBSpatial{}.

Given that well-defined \emph{formal topological relationships} and human-readable \emph{named topological relationship} queries are widely used in various scenarios~\cite{de9im}, ensuring their correctness is crucial. As far as we know, no general testing approach exists to identify logic bugs in queries regarding topological relationships, despite their importance.

\subsection{Affine Transformations \label{sec:affine-trans}}
We explore a fundamental concept in mathematics and computer graphics known as \emph{affine transformation}. 
Our intuition is that \emph{affine transformations} maintain topological relationships, because they preserve topological properties within \emph{Euclidean spaces}~\cite{affine-transformation}.
% In the \emph{Euclidean spaces}, 
Figure~\ref{fig:affine} presents intuitive examples of \emph{affine transformations}, including \emph{rotation}, \emph{translation}, \emph{scaling}, and \emph{shearing}. 
The checkerboard patterns, differentiated by high and low transparency, illustrate the geometries before and after the transformations, respectively. 

\begin{figure}[!t]
\centering
\includegraphics[width=.65\textwidth]{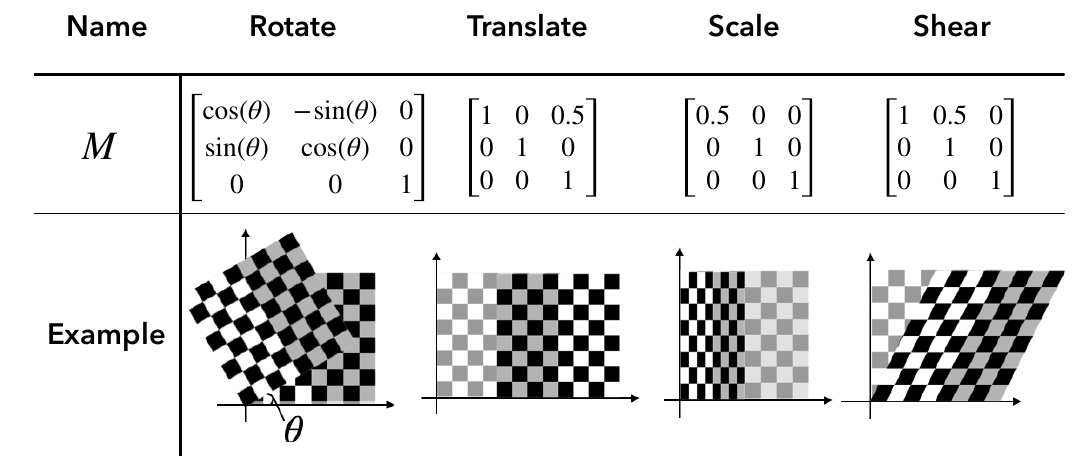}
\captionsetup{font={small, bf}}
\caption{ Examples of affine transformation. \label{fig:affine} }
\end{figure}

The widely-used topological relationship queries are implemented in \emph{Euclidean spaces} $\mathbb{R}^{2}$ and $\mathbb{R}^{3}$ in \SDBMS{}.
An affine transformation $\mathcal{A}: \mathbb{R}^{2} \rightarrow \mathbb{R}^{2}$ for a point $p \in \mathbb{R}^{2}$ is defined as 
\begin{equation}
    \mathcal{A}(p) = \mathbf{A}p + \mathbf{b} = 
    \begin{bmatrix}
    a_{11} & a_{12} \\
    a_{21} & a_{22} 
    \end{bmatrix} \cdot 
    \begin{bmatrix}
    x \\
    y
    \end{bmatrix} + 
    \begin{bmatrix}
    b_{1} \\
    b_{2}
    \end{bmatrix}
    =
    \begin{bmatrix}
    x' \\
    y'
    \end{bmatrix},
    \label{eq:affine}
\end{equation}
and $\mathcal{A}: \mathbb{R}^{3} \rightarrow \mathbb{R}^{3}$ for any point $p \in \mathbb{R}^{3}$ is defined as 
\begin{equation}
    \mathcal{A}(p) = \mathbf{A}p + \mathbf{b} = 
    \begin{bmatrix}
    a_{11} & a_{12} & a_{13}\\
    a_{21} & a_{22} & a_{23}\\ 
    a_{31} & a_{32} & a_{33} 
    \end{bmatrix} \cdot 
    \begin{bmatrix}
    x \\
    y \\
    z
    \end{bmatrix} + 
    \begin{bmatrix}
    b_{1} \\
    b_{2} \\
    b_{3}
    \end{bmatrix}
    =
    \begin{bmatrix}
    x' \\
    y' \\
    z' 
    \end{bmatrix},
    \label{eq:affine-3D}
\end{equation}
where $\mathbf{A}$ denotes an invertible matrix that describe the linear transformation; and $\mathbf{b}$ presents the translate vector.

The mapping matrix $M$ in Figure~\ref{fig:affine} is an augmented matrix that represents both the linear transformation and translation. 
$M$ has one more row and column than \text{Equation (\ref{eq:affine}) and (\ref{eq:affine-3D})}:
\begin{equation}
    \alpha'
    = M 
    \alpha
    = 
    \begin{bmatrix}
    \mathbf{A} & \mathbf{b} \\
    \mathbf{0} & 1 
    \end{bmatrix} 
    \alpha,
\label{eq:augmented-matrix}
\end{equation}
where $\alpha$ is the homogeneous vector of point $p$.

\section{Affine Equivalent Inputs\label{sec:AEI}}
In this section, we present the methodology to construct the expected result.
First, we formalize our intuition that \emph{topological relationships} remain preserved after applying \emph{affine transformations}. 
Then, we provide a proof of this proposition.
Based on the proved proposition,
we define \emph{Affine Equivalent Inputs} (\APPROACHAKA{}).

\paragraph{\textbf{Proof of our intuition}}

To formalize our intuition, we propose Proposition~\ref{proposition}, asserting that the topological relationship between a geometry pair is equal to that between its \emph{affine equivalent} pair. 
We use \emph{formal topological relationships} definition in Equation~\ref{eq:de9im}, since \emph{named topological relationships} are derived from them.
In \emph{formal topological relationships}, geometries are characterized by \emph{cell complexes}.
Given that the queries regarding topological relationships for geometry types are executed within the \emph{Euclidean spaces} $\mathbb{R}^2$ and $\mathbb{R}^3$,
we contextualize these \emph{cell complexes} accordingly. 

We introduce the concept of \emph{Euclidean spaces}, denoted as $\mathbb{R}^n$ where $2 \leq n \leq 3$, 
as a framework for delineating the structure of cell complexes. 
A geometry can be defined by a cell complex composed of cells.
Therefore, prior to defining affine equivalent geometries, we first define affine transformations on cells. 
A 0-cell (\emph{i.e.}, a point) in space $\mathbb{R}^2$ and $\mathbb{R}^3$ can be represented using Euclidean coordinates $(x, y)$ and $(x, y, z)$ respectively. 
A 1-cell, corresponding to a line segment, is characterized by the coordinates of its two distinct endpoints; 
a 2-cell, representing a triangle, is defined by three non-collinear edges;
a 3-cell, or a tetrahedron, consists of four triangles, forming a three-dimensional figure.
In $\mathbb{R}^2$ and $\mathbb{R}^3$, 2-cells and 3-cells are the highest-dimension cells, respectively.
Note that the geometry discussed here does not include curves.
Hence, an $n$-cell can be described by specifying the coordinates of its vertices.
Next, we define the affine transformation of cells in Definition~\ref{def:affine-cells}. 

\begin{definition}[Affine Transformations on Cells]
     For the vertices $\{p_1,...,p_n\}$ of a cell, an affine transformation can be defined as $\mathcal{A}: p_i \to {p_i}'$, where $1 \leq i \leq n$.
     \label{def:affine-cells}
\end{definition}

To formalize the concept of \emph{affine equivalent geometries}, we define an \emph{affine equivalence triplet} $(g, g', \mathcal{A})$, where $g$ and $g'$ are geometries, and $\mathcal{A}$ represents an affine transformation. 
We say that $g$ and $g'$ are affine equivalent under $\mathcal{A}$ if two conditions are both satisfied:
(1) for any cell $c$ in $g$, there exists a corresponding cell $c'$ in $g'$ such that $\mathcal{A}$ maps $c$ to $c'$,
and (2) for any cell $c'$ in $g'$, there exists a corresponding cell $c$ in $g$ such that the inverse of $\mathcal{A}$ maps $c'$ back to $c$.
Specifically, we apply an affine transformation to a geometry to generate a pair of \emph{affine equivalent geometries}. We further clarify the definition of affine equivalent geometry pairs in Definition ~\ref{def:AEGP}. 
Subsequently, in Proposition~\ref{proposition}, we formalize and prove our premise that the formal topological relationships remain invariant in pairs of affine equivalent geometries.

\begin{definition}[Affine Equivalent Geometry Pairs]
    Considering two geometry ordered pairs $(g_1, g_2)$ and $(g_1', g_2')$, they are affine equivalent geometry pairs if two affine equivalence triplets $(g_1, g_1', A)$ and $(g_2, g_2', A)$ exist.
    \label{def:AEGP}
\end{definition}

\begin{proposition}
For a geometry pair $(g_1, g_2)$, and its affine equivalent geometry pair $(g_1',g_2')$, we have $R(g_1, g_2) = R(g_1', g_2')$.
\label{proposition}
\end{proposition}

\begin{proof}
For brevity, we show the correctness of the relation between \emph{Boundary} of $g_1$ and \emph{Boundary} of $g_2$, \emph{i.e.}, $\mathcal{D}[\partial g_1 \cap \partial g_2] = \mathcal{D}[\partial g_1' \cap \partial g_2']$, while omitting the proof for remaining relations in $R$ that can be proven similarly.
To this end, we carefully consider the only two possible intersection cases for $\partial g_1$ and $\partial g_2$: \emph{non-empty} and \emph{empty} set.

\begin{enumerate}
\item \emph{\textbf{Non-empty set.}} 
In this case, the intersection of $\partial g_1$ and $\partial g_2$ is a non-empty set. 
Suppose $k$ cells in the set, denoted as $\partial g_1 \cap \partial g_2 = \sigma : \{ \sigma_1, \ldots, \sigma_k\}$.
% As the intersection is a non-empty set, we can assume that $\partial g_1 \cap \partial g_2 = \{ \sigma_1,...,\sigma_x\}$.
According to the Definition~\ref{def:AEGP}, for any cell $\sigma_i (1 \leq i \leq k)$, there exists a cell $\sigma'_i$ belonging to both $\partial g_1'$ and $\partial g_2'$, which implies $\sigma'_i \in \partial g_1' \cap \partial g_2'$ and thus $\mathcal{A}: \sigma \to \sigma'$ holds. 
Similarly, we have $\mathcal{A}^{-1}: \sigma' \to \sigma$.
Therefore, $\partial g_1 \cap \partial g_2$ is bijective to $\partial g_1' \cap \partial g_2'$. 
Considering that \emph{affine transformation} preserves the dimension~\cite{affine-transformation}, $\mathcal{D}[\partial g_1 \cap \partial g_2] = \mathcal{D}[\partial g_1' \cap \partial g_2']$ is thus proved in the non-empty set case.

\item \emph{\textbf{Empty set.}} 
In this case, the intersection of $\partial g_1$ and $\partial g_2$ is an empty set. 
We prove that the intersection of $\partial g_1'$ and $\partial g_2'$ is also an empty set by contradiction.
Assuming $\partial g_1' \cap \partial g_2' \neq \emptyset$, there must exist a face $\sigma_p$ in $\partial g_1' \cap \partial g_2'$.
According to the definition of affine equivalent geometry pairs, we can also find a face $\sigma'_p \in \partial g_1 \cap \partial g_2$ by $\mathcal{A}^{-1}: \sigma' \to \sigma$, which contradicts $\partial g_1 \cap \partial g_2 = \emptyset$.
\end{enumerate}
% In the same way, we can prove $R_{ij} = R_{ij}'$ for any $0 < i, j < 3$. Therefore, $R(g_1, g_2) = R'(g_1, g_2)$ is proved.
\end{proof}

\paragraph{\textbf{Definition of \APPROACHAKA{}}} 
We define a triplet $I = (\alpha,\beta,Q)$ as an input for an \SDBMS{}, where $\alpha$ and $\beta$ are two geometries, and $Q$ is a topological relationship query on geometry pair $\alpha$ and $\beta$.

\begin{definition}[\APPROACHNAME{}]
    $I_1 = (\alpha_1,\beta_1,Q_1)$ and $I_2 = (\alpha_2,\beta_2,Q_2)$ are \APPROACHNAME{}, if geometry pairs $(\alpha_1,\beta_1)$ and $(\alpha_2,\beta_2)$ are affine equivalent and $Q_1$ is equal to $Q_2$.
\label{def:AEI}
\end{definition}

According to Proposition~\ref{proposition}, passing \APPROACHNAME{} to the same \SDBMS{} should cause the topological relationship query to return the same result, unless the \SDBMS{} is affected by a bug.

\section{\ToolName{}}

%In this paper, we aim to propose an effective and efficient technique to detect logic bugs in queries regarding topological relationships. 
We propose \ToolName{} (\underline{Spat}ial \DBMS{}s Tes\underline{ter}), an automated testing tool and approach that combines \APPROACHAKA{} with a \SMARTGENERATOR{}.
The core idea behind \APPROACHAKA{} is to construct two \emph{affine equivalent} geometry pairs and check the consistency of their topological relationships. 
The \SMARTGENERATOR{} introduces a new concept for efficiently generating interesting test inputs,
which not only generates random geometries (using a \genbasedstrgy{}), but also derives geometries by applying spatial functions to existing geometries (using a \derivativestrgy{}).
The key insight behind the \SMARTGENERATOR{} is its ability to generate various topological relationships with a limited number of geometries, thereby enabling \APPROACHAKA{} to validate the results efficiently.

Figure~\ref{fig:overview} shows an overview of \ToolName{}.
Initially, we create a database \SDB{}1
with $N$ geometries using the \SMARTGENERATOR{}'s \emph{random-shape} and \emph{derivative} strategies (see Step \textcircled{1}); the \derivativestrgy{} creates geometries based on existing ones.
In the next step, we canonicalize each geometry in \SDB{}1 by obtaining an equivalent, potentially equivalent representation of the geometry, and then apply an affine transformation to construct a new geometry, resulting in \SDB{}2 (see Step \textcircled{2}).
Since \emph{affine transformations} preserve the topological relationships and \emph{canonicalization} produces equivalent geometries at the spatial level, the set of topological relationships of \SDB{}1 is equivalent to that of \SDB{}2, which we check for result validation. 
For result validation, 
we randomly fill the placeholders of a query; 
then we check whether the row counts retrieved from the same \SDBMS{} by this query are consistent between \SDB{}1 and \SDB{}2.
Inconsistent counts indicate a logic bug (see Step \textcircled{3}).

\definecolor{c1}{RGB}{42,99,172} % META
\definecolor{c2}{RGB}{255,88,93}
% \definecolor{c3}{RGB}{255,181,73}
\definecolor{c3}{RGB}{208,167,39}
\definecolor{c4}{RGB}{119,71,64} % Fast3D
\definecolor{c5}{RGB}{228,123,121} %Basic2D
\definecolor{c6}{RGB}{175,171,172} % Gray
% \definecolor{c6}{RGB}{136,88,180}
\definecolor{c7}{RGB}{0,51,153}
\definecolor{c8}{RGB}{56,140,139} % Fast2D

\definecolor{c9}{RGB}{0,0,0} %black
% \definecolor{c10}{RGB}{140,138,185} %purple
\definecolor{c10}{RGB}{120,80,190} % 蓝紫色

\definecolor{c11}{RGB}{255,204,0}
\definecolor{c12}{RGB}{128,128,128}
\definecolor{c13}{RGB}{98,148,96}
%sci-color
\definecolor{c14}{RGB}{184,168,207}
\definecolor{c15}{RGB}{253,207,158}
\definecolor{c16}{RGB}{182,118,108}
\definecolor{c17}{RGB}{175,171,172}

\definecolor{color1}{RGB}{0, 120, 190} % Teal Blue
\definecolor{color3}{RGB}{200, 180, 60} % Muted Gold

\definecolor{color2}{RGB}{220, 50, 50} % Vibrant Red
% Test for running time of different variants 

\definecolor{p1}{RGB}{174,223,172} % green
\definecolor{p2}{RGB}{224,175,107}  % orangle
\definecolor{p3}{RGB}{138,170,214}  % light blue
\definecolor{p4}{RGB}{222,117,123} %  % red
\definecolor{p5}{RGB}{216,174,174} % light pink
\definecolor{p6}{RGB}{163,137,214} % purple
\definecolor{p7}{RGB}{248,199,1} %  % yellow
\definecolor{p8}{RGB}{205,205,205} % grey
\definecolor{p9}{RGB}{255,0,127} % purple

\definecolor{t1}{RGB}{148, 190, 146} % 深绿色
\definecolor{t2}{RGB}{190, 149, 91} % 深橙色
\definecolor{t3}{RGB}{0, 153, 255}
\definecolor{t4}{RGB}{255, 0, 51}

\newsavebox{\mycodebox} 

% \begin{lrbox}{\mycodebox} 
% \begin{minipage}{0.34\textwidth}
% \vspace{-15pt}
%     \lstset{style=sqlstyle, numbers=none, xleftmargin=0em, framexleftmargin=0em}
%     \begin{lstlisting}[caption={Query Template (Instantiated as $q$)}, label={lst:query-template}]
% SELECT COUNT(*) FROM <table1> JOIN <table2> 
%   ON <TopoRlt>; 
%     \end{lstlisting}
% \vspace{-15pt}
% \end{minipage}
% \end{lrbox}

\begin{figure*}[t]

\includegraphics[width=0.9\textwidth]{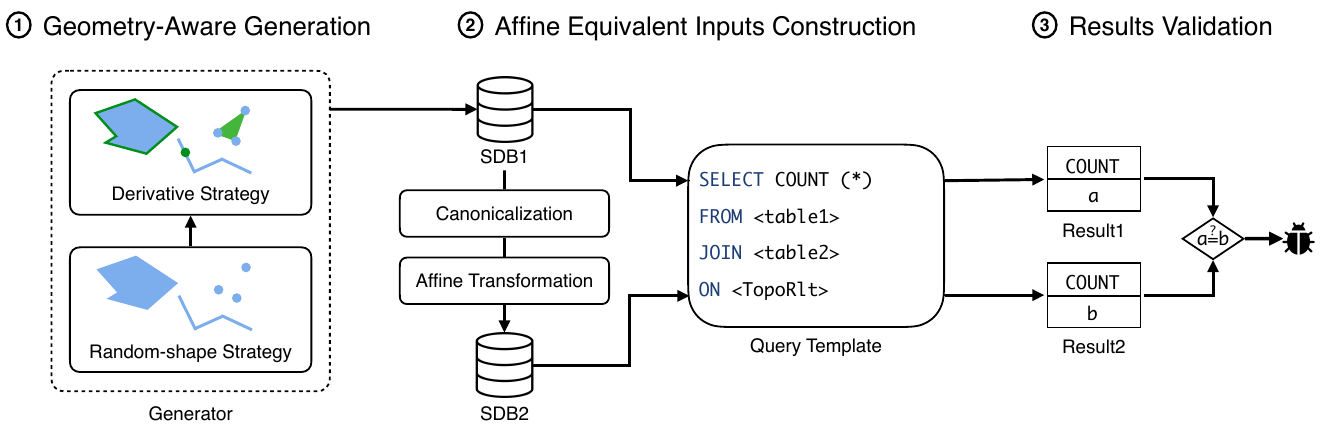}

\caption{Overview of \ToolName{}\label{fig:overview}}

\end{figure*}

\subsection{\SMARTGENERATORUPPER}
\label{sec:smart-gen}
In this section, we introduce a crucial component of \ToolName{}: a generator designed for efficiently detecting bugs in \SDBMS{}s. 
Previously, efforts to find logic bugs in SDBMSs relied on unit tests and user reports, as no automated geometry generators were designed for this purpose.
One naive approach is to randomly generate syntactically valid geometries, corresponding to our \genbasedstrgy{}.
However, the \genbasedstrgy{} makes it unlikely to observe a variety of topological relationships for the generated geometries, making it difficult to exercise the \SDBMS{}s. 
This issue is further exacerbated by the necessity of setting a limit on the number of geometries to avoid excessive time overheads in \SDBMS{}s (see Section~\ref{sec:time-distribution}).
To improve the efficiency of bug finding, we believe that a wider range of topological relationships within this limited number of geometries is required, allowing us to detect more logic bugs in queries regarding topological relationships. 
Therefore, we propose the \derivativestrgy{} that derives existing geometries by applying spatial functions.
The \SMARTGENERATOR{} employs two strategies for generating geometries: the \emph{random-shape} and \emph{derivative} strategies.

{\small
\begin{algorithm}[t]
\caption{\SMARTGENERATORUPPER}
\label{alg:gag}
\begin{algorithmic}[1]
\Require $N$, geometry number; $m$, table number 
\Ensure $sdb$, generated spatial database 
\Function{\textnormal{Generate}}{$N$, $m$}
    \State $sdb \gets$ \Call{CreateTables}{$m$} \label{line:createtables}
    \State $g \gets$ \Call{RandomShape}{} \label{line:init-randomshape}
    \State $sdb \gets$ \Call{InsertIntoRandomTable}{$sdb$, $g$}
    \For{$\_ \gets 2$ to $N$} \label{line:gag-begin}
        \If{\Call{RandomShape}{} is True} 
            \State $g \gets$ \Call{RandomShape}{}
        \Else{}
            \State $editFunc \gets$ \Call{RandomEditFunction}{}
            \State $g \gets$ \Call{Derive}{$sdb$, $editFunc$}
        \EndIf
        \State $sdb \gets$ \Call{InsertIntoRandomTable}{$sdb$, $g$}
    \EndFor \label{line:gag-end}
    \State \Return $sdb$
\EndFunction 

\Function{\textnormal{RandomShape}}{ } \label{line:randomshape-begin}
    \State $gType \gets$ \Call{RandomGeometryType}{} \label{line:geometry-type-select}
    \State $g \gets$ \Call{SyntaxGen}{gType} \label{line:syntax-gen}
    \State \Return $g$
\EndFunction \label{line:randomshape-end}

\Function{\textnormal{Derive}}{$sdb$, $editFunc$} \label{line:derive-begin}
    \State $k \gets$ the geometry number $editFunc$ needed \label{line:get-k}
    \State $gVector \gets$ k geometries randomly selected from $sdb$ \label{line:gVector}
    \State $g \gets$ \Call{editFunc}{$gVector$} \label{line:derive-sdbms}
    \If{Deriving $g$ failed} \label{line:checking-begin}
        \State $g \gets$ \Call{EmptyShape}{}
        \EndIf \label{line:checking-end}
    \State \Return $g$
\EndFunction \label{line:derive-end}

\end{algorithmic}
\end{algorithm}
}

\paragraph{\textbf{Random-shape strategy}}
The \genbasedstrgy{} generates syntactically valid geometries. % based on spatial data syntax.
Line~\ref{line:randomshape-begin}--\ref{line:randomshape-end} in Algorithm~\ref{alg:gag} demonstrate how the \genbasedstrgy{} generates a geometry: 
first, it randomly selects a geometry type $gType$ (Line~\ref{line:geometry-type-select}), 
and then uses the syntax definition for $gType$ to generate the geometry $g$ (Line~\ref{line:syntax-gen}). 
Geometries created using the \genbasedstrgy{} are valid at the syntax level, but not necessarily at the semantic level.
For example, the polygon \sqlinline{POLYGON((0 0,1 1,0 1,1 0,0 0))} is an invalid shape due to its self-intersecting boundaries, but it can be generated by the \genbasedstrgy{} since it is syntactically valid.
When a geometry is semantically invalid, the \SDBMS{} will indicate this with an error, which we ignore in Spatter.

\paragraph{\textbf{Derivative strategy}}
The \derivativestrgy{} is designed to create more topological relationships by deriving a geometry based on a set of editing functions provided by \SDBMS{}s.
Line~\ref{line:derive-begin}--\ref{line:derive-end} in Algorithm~\ref{alg:gag} illustrate how the \derivativestrgy{} generates a geometry.
First, we obtain the input number $k$ of editing function $editFunc$ (Line~\ref{line:get-k}).
Then, $k$ geometries are randomly selected from $sdb$ to construct a geometry vector $gVector$ (Line~\ref{line:gVector}).
The new geometry $g$ is derived from $gVector$ by $editFunc$ (Line~\ref{line:derive-sdbms}).
However, the derivation process may encounter errors if the $editFunc$ does not apply to the geometries in $gVector$.
To address this, we check for any failures, opting to generate an \emph{EMPTY} geometry in case of failure (Line~\ref{line:checking-begin}--\ref{line:checking-end}).
Finally, it returns geometry $g$ derived from the existing $k$ geometries. 

Table~\ref{tab:generator} shows editing functions categorized by geometric dimensions of input geometries. 
Generic functions can handle geometries of any dimension.  
In contrast, dimensional-based functions such as line-based, polygon-based, and multi-dimensional functions, operate exclusively on lines, polygons, and \emph{MULTI} and \emph{MIXED} geometries, respectively.

\begin{table*}[t]
\tiny
    \centering
    \captionsetup{font={small}}
    \caption{Categories of the operations in the \derivativestrgy{}.}
    \begin{tabular}{lll}
      \toprule
       Type & Example & Description \\
      \midrule
       \multirow{2}{*}{Line-Based}  
       & SetPoint & Replace a specific point of an input LINESTRING with a given point.\\
       & Polygonize & Create a GEOMETRYCOLLECTION containing the polygons formed by the line.\\
       \midrule
       \multirow{2}{*}{Polygon-Based}  & DumpRings & Extract the rings of an input POLYGON. \\
       & ForcePolygonCW & Force an input POLYGON or MULTIPOLYGON to use a clockwise orientation for their exterior ring. \\
       \midrule
       \multirow{2}{*}{Multi-Dimensional}
       & GeometryN & Fetch the Nth geometry element from an input MULTI or MIXED geometry based on 1-based indexing.\\
       & CollectionExtract & Produce a collection of geometries of a specified type, extracted from an input MULTI or MIXED geometry.\\
       \midrule
       \multirow{2}{*}{Generic} 
       & Boundary & Retrieve the boundary of an input geometry. \\ 
       & ConvexHull & Generate the convex hull of an input geometry. \\
      \bottomrule
    \end{tabular}
    \label{tab:generator}
\end{table*}

\paragraph{\textbf{Generation process}} 
The $Generate$ function in Algorithm~\ref{alg:gag} outlines the generation process of the \SMARTGENERATOR{}. 
The algorithm expects a specified number of geometries, $N$, and a target number of tables, $m$, as input.
During the initialization phase, a spatial database $sdb$ is created with $m$ empty tables (Line~\ref{line:createtables}).
Since no geometries can be derived from an empty table,
we employ the \genbasedstrgy{} to generate a geometry and randomly insert it into one of $sdb$ tables (Line~\ref{line:init-randomshape}).
For the creation of each geometry $g$, we randomly select one strategy and insert the generated geometry into a randomly chosen table (Line~\ref{line:gag-begin}--\ref{line:gag-end}).
Last, we return the newly-created spatial database $sdb$.

\subsection{\textbf{Affine Transformation}}
Proposition~\ref{proposition} suggests a high-level idea: \emph{affine transformations} preserve topological relationships. Therefore, we apply the same affine transformation to each geometry in the generated database.

Algorithm~\ref{alg:affine} shows the workflow of constructing an \emph{affine equivalent} geometry for $g$. 
The inputs of this workflow are the geometry $g$ and the \emph{Euclidean space} dimension $n$. 
We first generate a mapping matrix $M$ (Line~\ref{line:generateM}). 
The $GenerateMappingMatrix$ function involves generating an invertible matrix $A$ and a translation vector $b$  (Line~\ref{line:genrateM-start}--\ref{line:genrateM-end}).
Then, we create $g'$, a copy of the geometry $g$ (Line~\ref{line:copy}). 
For each point $p$ in $g'$, we update its coordinates with the values generated by the \emph{Affine} function (Line~\ref{line:point-start}--\ref{line:point-end}). 
The \emph{Affine} function performs three steps:
converting the coordinate point $p$ into a homogeneous vector, 
left-multiplying the vector by matrix $M$, referring to \text{Equation (\ref{eq:augmented-matrix})},
and transforming the resulting vector back into a coordinate point.
Finally, the algorithm returns $g'$, which is an affine equivalent geometry of $g$ (Line~\ref{line:return}).

{\small
\begin{algorithm}[t]
\caption{Apply an affine transformation on a geometry.}
\label{alg:affine}
\begin{algorithmic}[1]
\Require  $g$, a geometry; $n$, Euclidean space dimension
\Ensure $g'$, affine equivalent input of $g$
\Function{\textnormal{Construct}}{$g$, $n$}
    \State $M \gets \Call{GenerateMappingMatrix}{n}$ \label{line:generateM}
    \State $g' \gets g$ \label{line:copy}
    \For{\textbf{each} $p$ \textbf{in} $g'$} \label{line:point-start}
        \State $p \gets \Call{Affine}{p, M}$\label{line:point-end}
    \EndFor

    \State \Return $g'$ \label{line:return}
\EndFunction 

\Function{\textnormal{GenerateMappingMatrix}}{$n$} \label{line:genrateM-start}
    \State $A \gets$ a random non-singular matrix $\in \mathbb{R}^{n \times n}$ \label{line:generateA}
    \State $b \gets$ a random vector $\in \mathbb{R}^{n}$ \label{line:generateb}
    \State $M \gets \Call{AugmentedMatrix}{A,b} $
    \State \Return $M$ \label{line:genrateM-end}
\EndFunction 

\end{algorithmic}
\end{algorithm}
}

\paragraph{\textbf{Avoiding precision issues}} 
An affine transformation utilizes a matrix to construct affine equivalent geometries. 
However, using matrices with floating-point numbers introduces precision issues.   
For example, consider \text{Equation (\ref{eq:imprecision-example})}, where we apply an affine transformation to a 2D point represented by \sqlinline{POINT(0.1,0)}. We first convert the point into a homogeneous vector $p$, then generate a matrix $M$ to left-multiply $p$, obtaining a vector $p'$, and finally convert it back to a point with coordinates $(0.02 + 4 \times 10^{-18}, 0)$. 
Precision issues arise because multiplying floating-point numbers on computers yields values that are very close to, but not exactly, the theoretical values.
Specifically, while in theory, $0.1$ multiplied by $0.2$ equals $0.02$, the representation of floating-point numbers in computers may lead to an imprecise value, with a potential error of \sqlinline{4E-18}.
{\small
\begin{equation}
    p' = M \cdot p
    =
    \begin{bmatrix}
        0.2 & 1 & 0\\ 0 & 1 & 0\\ 0 & 0 & 1
    \end{bmatrix}  
    \begin{bmatrix} 0.1 \\ 0 \\ 1 \end{bmatrix} = 
    \begin{bmatrix} 0.02 + 4 \times 10^{-18} \\ 0 \\ 1 \end{bmatrix}
\label{eq:imprecision-example}
\end{equation}
}

Although such errors can be circumvented at the implementation level during affine transformations, precision issues with floating-point numbers still occur in \SDBMS{}s.
The bug shown in Listing~\ref{lst:precision-issue} illustrates an example of precision issues in \PostGIS{}.
Despite this bug being confirmed by the developers at the beginning of our testing campaign, 
bugs unrelated to precision issues are more eagerly anticipated by both the developers and us.
To avoid false alarms caused by precision issues, we chose not to introduce floating-point numbers during the steps of affine transformations and spatial data generation. 
Specifically, for an affine transformation, we generate matrix $A$ (Line~\ref{line:generateA}) and vector $b$ (Line~\ref{line:generateb}) using random integers, as outlined in Algorithm~\ref{alg:affine}.

\subsection{\textbf{Canonicalization}\label{sec:canonical}}
We consider canonicalization as a special case of constructing affine equivalent geometries, since applying the special matrix $\mathbf{E}$ produces the geometry itself. 
In canonicalization, we aim to transform the original representation of each geometry into a canonical representation, which not only constructs another kind of expected results effective for testing, but also provides a pre-processing method for Algorithm~\ref{alg:affine}. 
We define operations to convert the representation of each geometry to a canonical one. 
We convert the representations at the element and value level.

\begin{figure}[t]
% \begin{tikzpicture}[node distance=0.3cm and 0.3cm, auto, align=left]
%     \tikzstyle{box} = [rectangle, draw, font=\footnotesize, minimum width=1cm, align=left]
%     \tikzstyle{dashedline} = [dashed, line width=1pt, dash pattern=on 4pt off 3pt]
    
%     \node[] (o) {\textcolor{p8}{MULTI}\\LINESTRING\textcolor{p8}{(}\\(0 2,1 0,3 1,\textcolor{p8}{3 1,}5 0)\\\textcolor{p8}{,EMPTY)}};

%     \node[dashedline, below=0.1cm of o, xshift=6cm] (dl1) {------------------------------------------------------------------------------------------------------------------------};
    
%     \node[below=0.1 of dl1] (e-level) {Element Level};
%         \node[box, below=0cm of e-level] (homo) {Homogenization};
%         \node[box, right= of homo] (empty) {\sqlinline{EMPTY} Removal};
%         \node[box, right= of empty] (e-dup) {Duplicate Removal};
%         \node[box, right= of e-dup] (e-re) {Reorder};

%         \draw[-latex] (homo) -- (empty);
%         \draw[-latex] (empty) -- (e-dup);
%         \draw[-latex] (e-dup) -- (e-re);
        
%     \node[below=0.6 of e-level, xshift=-0.2cm] (v-level) {Value Level};
%         \node[box, below=0cm of v-level, xshift=1cm] (v-dup) {Consecutive Duplicate Removal};
%         \node[box, right= of v-dup] (v-re) {Reorder};
       
%         \draw[-latex] (v-dup) -- (v-re);
    
% \end{tikzpicture}
\begin{tikzpicture}[node distance=1cm, auto, font=\scriptsize\sffamily]
    \tikzstyle{box} = [rectangle, draw, font=\scriptsize\sffamily, minimum width=1cm, align=left]

  \node (des1) [yshift=0.5cm, xshift=4cm] {\textcolor{p8}{MULTI}LINESTRING\textcolor{p8}{(}(0 2,1 0,3 1,3 1,5 0)\textcolor{p8}{,EMPTY)}};
  \draw [-, dashed] (0,0) -- (8.3,0);
  \node (e-level) [yshift=-0.2cm, xshift=1cm] {Element Level};
  \draw [-, dashed] (0,-1) -- (8.3,-1);
  \node (des2) [yshift=-1.5cm, xshift=4cm] {LINESTRING(0 2,1 0,3 1,\textcolor{p8}{3 1,}5 0)};
  \draw [-, dashed] (0,-2) -- (8.3,-2);
  \node (v-level) [yshift=-2.2cm, xshift=0.85cm] {Value Level};
  \draw [-, dashed] (0,-3) -- (8.3,-3);
  \node (des3) [yshift=-3.5cm, xshift=4cm] {LINESTRING(0 2,1 0,3 1,5 0)};
 
  \draw [-latex] (des1) -- (4,0);
  \draw [-latex] (4,-1) -- (des2);
  \draw [-latex] (des2) -- (4,-2);
  \draw [-latex] (4,-3) -- (des3);
  
  \node (elem1) [box, below=0.1 of e-level] {\CapEmptyRemoval};
  \node (elem2) [box, right=0.3 of elem1] {\CapHomogenization};
  \node (elem3) [box, right=0.3 of elem2] {\CapDuplicateRemoval};
  \node (elem4) [box, right=0.3 of elem3] {\CapReordering};

  \draw [-latex] (elem1) -- (elem2);
  \draw [-latex] (elem2) -- (elem3);
  \draw [-latex] (elem3) -- (elem4);

  \node(value1) [box, below=0.05 of v-level, xshift=0.9cm] {\CapConsecutiveDuplicateRemoval};
  \node(value2) [box, right=0.3 of value1] {\CapReordering};
  
  \draw [-latex] (value1) -- (value2);
\end{tikzpicture}
\captionsetup{font={small}}
\caption{Canonicalization for geometry at different levels.}
\label{fig:canonicalization}
\end{figure}
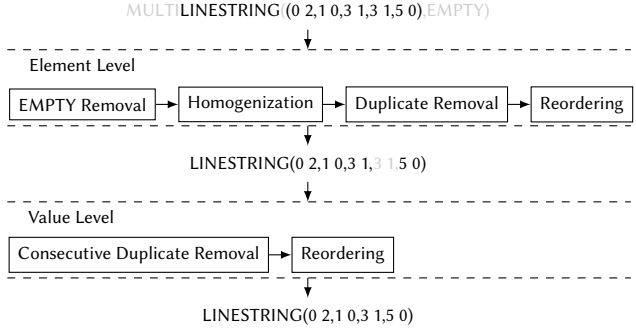

\paragraph{\textbf{Element level}\label{sec:element-level}} 
The concept of \emph{element} pertains to both MULTI and MIXED geometries.
Canonicalization at the element level involves four steps and applies exclusively to MULTI and MIXED geometries. 
The process begins with \EmptyRemoval{}, eliminating any \emph{EMPTY} elements present. 
Next, \emph{homogenization} transforms a MULTI geometry containing only a single element into a basic-type geometry 
and flattens nested collections into a uniform structural representation. 
Subsequently, duplicated elements are removed; note that duplicates are identified based on their shape.
Finally, the elements are reordered according to their dimensions.

Figure~\ref{fig:canonicalization} illustrates the canonicalization of a {\small MULTILINESTRING} at the element level. Initially, the \emph{EMPTY} element is removed, yielding a \sqlinline{MULTILINESTRING((2,1 0,3 1,3 1,5 0))}; 
subsequently, it is homogenized into a \sqlinline{LINESTRING(2,1 0,3 1,3 1,5 0)}, since there is only one element within the MULTI geometry.  
In this example, neither the \DuplicateRemoval{} nor the \Reordering{} is necessary.

\paragraph{\textbf{Value level}}
The value-level canonicalization process is designed for each basic-type element.
The location of a basic-type element is specified by an ordered number of pairs or triples.
Specifically, in 2D geometries, a {\small POINT} is represented by a pair of values; 
a {\small LINESTRING} is presented by ordered points; 
and a {\small POLYGON} is defined by one or multiple loops that ordered points that start and end at the same point.

Canonicalization at the value level involves two steps: \ConsecutiveDuplicateRemoval{} and \Reordering{}. 
First, we eliminate redundant points that are identical to the preceding point.
Then, we reorder the coordinates by direction. For a {\small LINESTRING}, we determine whether to reverse by comparing the values of the endpoints in the order of the x-axis, y-axis, and z-axis. 
For a {\small POLYGON}, we convert all the loops into clockwise orientation.

Figure~\ref{fig:canonicalization} demonstrates the canonicalization of a {\small LINESTRING} at the value level.
First, we remove the redundant point \sqlinline{(3 1)}, resulting in \sqlinline{LINESTRING(0 2,1 0,3 1,5 0)}. Then, we reorder the coordinates. Since the endpoint coordinates are ordered as expected, reordering does not alter the representation.

\subsection{\textbf{Results Validation}\label{sec:result-validation}} 
By utilizing \emph{canonicalization} and \emph{affine transformation}, it is possible to establish multiple sets of \APPROACHNAME{} (\APPROACHAKA{}).
Consider a spatial database named \SDB{}1. For each geometry $g$ within \SDB{}1, we perform canonicalization and apply an affine transformation to construct its \emph{affine equivalent} geometry, $g'$, resulting in a new spatial database, \SDB{}2.
The geometries $g$ and $g'$ are stored in tables of the same name in their respective databases.

Figure~\ref{fig:overview} presents a query template containing three placeholders: 
two table names (\emph{i.e.}, <table1>, <table2>) and one for a topological relationship condition (\emph{i.e.}, <TopoRlt>).
Two table names are valid and selected randomly from \SDB{}1.
The boolean condition <TopoRlt> is valid for the tested \SDBMS{} and is randomly selected from a list containing conditions regarding topological relationships from \SDBMS{} user manuals.  
When checking whether the results evaluated by the same query are consistent between \SDB{}1 and \SDB{}2,
any discrepancy reveals a bug in the tested \SDBMS{}s.
To illustrate, consider the case where t1 and t2 are two valid tables in \SDB{}1, and \sqlinline{ST_Covers} is a topological relationship function specific to \PostGIS{}.
The query filled with t1, t2, and \sqlinline{ST_Covers(t1.g, t2.g)} retrieves the row count from \PostGIS{}, indicating whether geometries in t1 cover geometries in t2.
If the count retrieved from \SDB{}1 is not equal to that from \SDB{}2, we detect a potential logic bug.

\section{Evaluation}

In this section, we first summarize and classify the detected bugs, 
along with feedback from developers. 
Then, we analyze the bugs found by \APPROACHAKA{}, present bug examples, and discuss bug-inducing patterns.
Next, we assess whether the previous methodologies could detect the logic identified by \APPROACHAKA{}.
Finally, we conduct experiments to demonstrate the efficiency of \ToolName{}.

\paragraph{Tested \SDBMS{}s}
We evaluated our approach on 4 well-known, mature, and actively maintained \SDBMS{}s: \PostGIS{}, \DuckDBSpatial{}, \MySQLGIS{}, and \SQLSERVER{}.
\PostGIS{} was chosen as our primary testing target due to its popularity and high ranking in the DB-Engines Ranking.
\DuckDB{} is recognized as the most popular embedded OLAP system~\cite{10.1145/3299869.3320212},
and we selected its extension, \DuckDBSpatial{}, for testing.
\MySQLGIS{}, one of the most popular open-source relational \DBMS{}s, was tested for its built-in spatial functionality. In addition, we included one commercial \SDBMS{}, namely \SQLSERVER{}. However, given limited feedback on our reported bugs, we discontinued our testing efforts and did not consider other commercial \SDBMS{}s.

\subsection{New Bugs\label{sec:new-bugs}}

\paragraph{Methodology.} 
To evaluate the effectiveness of \APPROACHNAME{} in detecting bugs,
we incrementally implemented \ToolName{} and intermittently tested the latest versions of \SDBMS{}s over four months in a testing campaign.
Typically, Spatter ran for 10 minutes to 1 hour, as it detected issues quickly.
\ToolName{} records two sequences of statements for generating the \emph{affine equivalent} databases for each discrepancy.
Before reporting any potential issues, 
we performed two steps. 
First, we both automatically~\cite{988498} and manually reduced them.
Second, we determined which result violated the definition of topological relationship functions according to the \SDBMS{}s' documentation. 

In testing \PostGIS{} and \DuckDBSpatial{}, 
we discovered that some bugs originated from a common third-party library named \GEOS{}.
At the beginning of our testing campaign, we reported potential bugs to the \PostGIS{} and \DuckDBSpatial{} communities. The developers then confirmed them as upstream bugs and reported them to the \GEOS{} community.
Subsequently, if we determined the bug was within their library, we directly reported bugs to \GEOS{}.
After fixes were applied, we updated the \GEOS{} version on the \SDBMS{}s on our machine to check if the issue persisted.
We reported an issue as a separate bug to the \SDBMS{} community only if persisted in the tested \SDBMS{} after all fixes were applied.

\paragraph{\textbf{Overall bug detection results}}
Table~\ref{tab:bug-summary} summarizes the bugs identified during our testing campaign. 
We classified the detected bugs into four distinct categories:
\begin{itemize}
    \item \emph{Fixed bugs} refer to those that have been confirmed by developers and subsequently addressed with fixing patches.
    \item \emph{Confirmed bugs} refer to the bugs that have been acknowledged by developers, but have not been fixed.
    \item \emph{Unconfirmed bugs} refer to the bugs that are identified manually according to the usage documentation and are awaiting developer confirmation. 
    \item \emph{Duplicate bugs} refer to bugs confirmed to have the same cause as previously confirmed bugs.
\end{itemize}
We consider \DTLfetch{status}{status}{Identified}{Sum} of them as previously unknown, unique bugs, \DTLfetch{status}{status}{Real}{Sum} of which have been confirmed or fixed by the developers.
Most of them affected \PostGIS{}, as its developers quickly responded and addressed many of the issues we reported.
We detected 2 bugs in \SQLSERVER{} and reported them to the \SQLSERVER{} community, 
but have not received any response; hence, we ceased testing it. 
One of the reported bugs was a duplicate, caused by the same root cause as a previously confirmed bug. 

{\small
\begin{table}
\centering
\captionsetup{font={small}}
\caption{Status of the reported bugs in \SDBMS{}s. \GEOS{} is a third-party library used by \PostGIS{} and \DuckDBSpatial{}. The bugs detected in \GEOS{} are listed separately.}
\begin{tabular}{l@{\hspace{2mm}}r@{\hspace{2mm}}r@{\hspace{2mm}}r@{\hspace{2mm}}c@{\hspace{2mm}}r}
    \toprule
    \textbf{\SDBMS{}} & \small{Fixed} & \small{Confirmed} & \small{Unconfirmed} & \small{Duplicate} & \small{Sum} \\
    \midrule
    \GEOS       
    & \DTLfetch{status}{status}{GEOS}{Fixed} 
    & \DTLfetch{status}{status}{GEOS}{Confirmed} 
    & \DTLfetch{status}{status}{GEOS}{Unconfirmed}
    & \DTLfetch{status}{status}{GEOS}{Duplicate} 
    & \DTLfetch{status}{status}{GEOS}{Sum} \\
    \PostGIS        
    & \DTLfetch{status}{status}{PostGIS}{Fixed} 
    & \DTLfetch{status}{status}{PostGIS}{Confirmed} 
    & \DTLfetch{status}{status}{PostGIS}{Unconfirmed}
    & \DTLfetch{status}{status}{PostGIS}{Duplicate} 
    & \DTLfetch{status}{status}{PostGIS}{Sum} \\
    \DuckDBSpatial      
    & \DTLfetch{status}{status}{DuckDBSpatial}{Fixed}  
    & \DTLfetch{status}{status}{DuckDBSpatial}{Confirmed} 
    & \DTLfetch{status}{status}{DuckDBSpatial}{Unconfirmed} 
    & \DTLfetch{status}{status}{DuckDBSpatial}{Duplicate} 
    & \DTLfetch{status}{status}{DuckDBSpatial}{Sum} \\
    \MySQLGIS       
    & \DTLfetch{status}{status}{MySQLGIS}{Fixed}  
    & \DTLfetch{status}{status}{MySQLGIS}{Confirmed} 
    & \DTLfetch{status}{status}{MySQLGIS}{Unconfirmed} 
    & \DTLfetch{status}{status}{MySQLGIS}{Duplicate} 
    & \DTLfetch{status}{status}{MySQLGIS}{Sum} \\
    \SQLSERVER      
    & \DTLfetch{status}{status}{SQLSERVER}{Fixed}  
    & \DTLfetch{status}{status}{SQLSERVER}{Confirmed} 
    & \DTLfetch{status}{status}{SQLSERVER}{Unconfirmed} 
    & \DTLfetch{status}{status}{SQLSERVER}{Duplicate} 
    & \DTLfetch{status}{status}{SQLSERVER}{Sum} \\
    \midrule
    Sum 
    & \DTLfetch{status}{status}{Sum}{Fixed}  
    & \DTLfetch{status}{status}{Sum}{Confirmed} 
    & \DTLfetch{status}{status}{Sum}{Unconfirmed} 
    & \DTLfetch{status}{status}{Sum}{Duplicate} 
    & \DTLfetch{status}{status}{Sum}{Sum} \\
    \bottomrule
\end{tabular}

\label{tab:bug-summary}
\end{table}
}

Table~\ref{tab:bugs-types} shows the previously unknown and confirmed bugs, 
which are classified into logic and crash bugs.
Out of the \DTLfetch{type}{type}{Sum}{Sum} bugs, the majority,
\DTLfetch{type}{type}{Logic}{Sum} bugs were logic bugs that we aimed to find. 
We detected 9 logic bugs in \GEOS{} and 7 \SDBMS{}-specific logic bugs in \PostGIS{}.
Besides, \ToolName{} detected \DTLfetch{type}{type}{Crash}{Sum} crash bugs.
All of the crash bugs were fixed within one day.
Logic bugs are more difficult to fix than crash bugs, 
because pinpointing the root cause is usually more difficult. 
At the same time, developers must ensure that patches do not affect other scenarios. 
For some bugs, 
\PostGIS{} may need to propose and implement a new algorithm to address them, 
as indicated by developer feedback: \emph{"There is a new topological predicate algorithm being worked on which solves this problem. It should be ready by mid-2024.
"}

{\small
\begin{table}
\centering
\captionsetup{font={small}}
\caption{A classification of the confirmed and fixed bugs. Note that 2 fixed logic bugs are not listed here, because they were found in \JTS{}, which is not an \SDBMS{}.}
\begin{tabular}{l@{\hspace{2mm}}r@{\hspace{2mm}}r@{\hspace{2mm}}r@{\hspace{2mm}}r@{\hspace{2mm}}r}
    \toprule
      & \multicolumn{2}{c}{\textbf{Logic Bugs}} &  \multicolumn{2}{c}{\textbf{Crash bugs}} & \\
      \textbf{\SDBMS{}} & \small{Fixed} & \small{Confirmed} & \small{Fixed} & \small{Confirmed} & \small{Sum} \\
    \midrule
    \GEOS     
    & \DTLfetch{type}{type}{GEOS}{LogicFixed} 
    & \DTLfetch{type}{type}{GEOS}{LogicConfirmed} 
    & \DTLfetch{type}{type}{GEOS}{CrashFixed} 
    & \DTLfetch{type}{type}{GEOS}{CrashConfirmed} 
    & \DTLfetch{type}{type}{GEOS}{Sum} \\
    \PostGIS     
    & \DTLfetch{type}{type}{PostGIS}{LogicFixed} 
    & \DTLfetch{type}{type}{PostGIS}{LogicConfirmed} 
    & \DTLfetch{type}{type}{PostGIS}{CrashFixed} 
    & \DTLfetch{type}{type}{PostGIS}{CrashConfirmed} 
    & \DTLfetch{type}{type}{PostGIS}{Sum} \\
    \MySQLGIS       
    & \DTLfetch{type}{type}{MySQLGIS}{LogicFixed} 
    & \DTLfetch{type}{type}{MySQLGIS}{LogicConfirmed} 
    & \DTLfetch{type}{type}{MySQLGIS}{CrashFixed} 
    & \DTLfetch{type}{type}{MySQLGIS}{CrashConfirmed} 
    & \DTLfetch{type}{type}{MySQLGIS}{Sum} \\
    \DuckDBSpatial      
    & \DTLfetch{type}{type}{DuckDBSpatial}{LogicFixed} 
    & \DTLfetch{type}{type}{DuckDBSpatial}{LogicConfirmed} 
    & \DTLfetch{type}{type}{DuckDBSpatial}{CrashFixed} 
    & \DTLfetch{type}{type}{DuckDBSpatial}{CrashConfirmed} 
    & \DTLfetch{type}{type}{DuckDBSpatial}{Sum} \\
    \midrule
    Sum
    & \DTLfetch{type}{type}{Sum}{LogicFixed} 
    & \DTLfetch{type}{type}{Sum}{LogicConfirmed} 
    & \DTLfetch{type}{type}{Sum}{CrashFixed} 
    & \DTLfetch{type}{type}{Sum}{CrashConfirmed} 
    & \DTLfetch{type}{type}{Sum}{Sum} \\
    \bottomrule
\end{tabular}

\label{tab:bugs-types}
\end{table}
}

\paragraph{Impact of \ToolName{}.}
Developer feedback is an important indicator of whether the testing work has a positive impact in practice. 
We received much positive feedback from  \PostGIS{} developers.
They noticed our efforts and reached out to us via email: 
\emph{"Congratulations on your thorough work in this area - it benefits all uses of PostGIS and GEOS!"}
They emphasized the significant contribution of our approach in detecting logic bugs: 
\emph{"I'm aware of some work around automated fuzzing tests for GEOS, but typically they just involve obscure input causing crashes, not simple test cases producing wrong answers."}
In fixing one of the \PostGIS{} bugs,
a developer noted that a test case revealed an incorrect definition. 
They acknowledged, \emph{"it has become clear to me that our definition of \lstinline{ST_DFullyWithin} is probably “wrong”, that is, it’s not what people think they are getting when they call it."}
From \DuckDBSpatial{}, we received numerous appreciations from one of the core developers, saying, \emph{"Thanks for reporting this issue!"} 
All the bugs were resolved and 4 of them were fixed within 24 hours. 
From \MySQLGIS{}, our reported bugs were confirmed by a testing developer within one day, 
one of which was marked as \emph{serious}. 
Besides, a bug was fixed and included in the subsequent release version.

\subsection{Illustrative Examples\label{sec:case-study}}
In this section, we discuss selected bugs to illustrate why \APPROACHAKA{} can detect them, whereas differential testing cannot.
We then present our results from manually analyzing all the bugs detected by \APPROACHAKA{} and determine how many of these bugs could also be identified by differential testing or TLP.
Lastly, we identify common patterns that induce bugs to provide insights for implementing \SDBMS{}s.

\paragraph{Selected bugs.} We selected a few examples of logic bugs identified by \ToolName{} to illustrate the effectiveness of \APPROACHAKA{} in bug detection, 
evidenced by the diversity of detected bugs.
The variety in their root causes further underscores this diversity.
For conciseness, we present only simplified test cases that highlight the underlying core issues, rather than the original test cases used to uncover these bugs.

\paragraph{\textbf{Incorrect result after scaling in \MySQLGIS{}}}
Listing~\ref{lst:scaling} illustrates a logic bug related to an incorrect relation calculation. \MySQLGIS{} incorrectly determines that \sqlinline{@g1} crosses \sqlinline{@g2}, 
violating the condition that 
\emph{"their intersection is not equal to either of the two given geometries"}. 
The results before and after scaling demonstrate this inconsistency. 
After identifying that the result containing an \emph{EMPTY} element was incorrect, we reported it to the \MySQLGIS{} community.

The above bug is a logic bug that is difficult to detect by differential testing,
as it is concealed within expected discrepancies.
The definition of the function \sqlinline{ST_Crosses} varies among different \SDBMS{}s, leading to expected discrepancies.

\begin{figure}[t]
\lstset{style=sqlstyle}
\begin{lstlisting}[caption={\MySQLGIS{} calculates an incorrect relation after scaling coordinates by a factor of 10.}, label={lst:scaling}]
SET @g1='MULTILINESTRING((99`\bugPattern{0}` 28`\bugPattern{0}`,10`\bugPattern{0}` 2`\bugPattern{0}`))';
SET @g2='GEOMETRYCOLLECTION(MULTILINESTRING((99`\bugPattern{0}` 28`\bugPattern{0}`,    10`\bugPattern{0}` 2`\bugPattern{0}`)),POLYGON((36`\bugPattern{0}` 6`\bugPattern{0}`,85`\bugPattern{0}` 62`\bugPattern{0}`,85`\bugPattern{0}` 42`\bugPattern{0}`,36`\bugPattern{0}` 6`\bugPattern{0}`)))';
SELECT ST_Crosses(ST_GeomFromText(@g1), ST_GeomFromText(@g2));
-- {1} `\faBug{}` {0} `\faCheckCircle{}` 
\end{lstlisting}
\end{figure}

\paragraph{\textbf{Incorrect result after swapping axes in \MySQLGIS{}}}
Listing~\ref{lst:overlaps} shows a confirmed logic bug in \MySQLGIS{}. 
The Function \sqlinline{ST_Overlaps(g1, g2)}  returns \sqlinline{1} if \sqlinline{g1} and \sqlinline{g2} intersect and their intersection results in geometry of the same dimension but not equal to either \sqlinline{g1} or \sqlinline{g2}; otherwise, it returns \sqlinline{0}.
In this case, since the intersection of \sqlinline{g1} and \sqlinline{g2} is equal to \sqlinline{g1}, 
the expected result of \sqlinline{ST_Overlaps(g1, g2)} should be \sqlinline{0}.

This logic bug cannot be detected by differential testing.
\PostGIS{} and \DuckDBSpatial{} consider \sqlinline{g2} invalid, because two elements of \sqlinline{g2} intersect, resulting in a self-intersection error.
This serves as an example of how different data designs can limit the effectiveness of differential testing approaches.

\begin{figure}[t]
\lstset{style=sqlstyle}
\begin{lstlisting}[caption={\MySQLGIS{} identifies an incorrect relation after swapping the X and Y axes.}, label={lst:overlaps}]
SET @g1 = ST_GeomFromText('POLYGON((614 445,30 26,80 30,614 445))');
SET @g2 = ST_GeomFromText('GEOMETRYCOLLECTION(POLYGON((614 445,30 26,80 30,614 445)),POLYGON((190 1010,40 90,90 40,190 1010)))');
SELECT ST_Overlaps(@g2, @g1);
-- {0} `\faCheckCircle{}`
SELECT ST_Overlaps(ST_SwapXY(@g2), ST_SwapXY(@g1));
-- {1} `\faBug{}` {0} `\faCheckCircle{}`
\end{lstlisting}
\end{figure}

\paragraph{\textbf{EMPTY-related bugs in \PostGIS{}}} 
Listing~\ref{lst:distance} shows a fixed logic bug in \PostGIS{}. 
\sqlinline{ST_Distance} calculates the distance between the given two geometries \sqlinline{g1} and \sqlinline{g2}. 
The function returns the minimum distance if the given geometries are \emph{MULTI geometries}. 
Thus, the expected result is $2$. 
However, \PostGIS{} returns $3$ incorrectly.
The developer repaired the incorrect recursive logic after we reported it.
This bug was found by \emph{canonicalization}.

This example shows a logic bug that can be missed by comparing \SDBMS{}s with and without supporting \sqlinline{EMPTY} elements.

\ignorelst{}{
\begin{figure}[t]
\lstset{style=sqlstyle}
\begin{lstlisting}[caption={\PostGIS{} calculates a distance incorrectly.}, label={lst:distance}]
SELECT ST_Distance('MULTIPOINT((1 0),(0 0))'::geometry,        'MULTIPOINT((-2 0), `\bugPattern{EMPTY}`)'::geometry); `\label{line:distance-origin}`
-- {3} `\faBug{}`  {2} `\faCheckCircle{}`
SELECT ST_Distance('MULTIPOINT((1 0),(0 0))'::geometry,        'POINT(-2 0)'::geometry); `\label{line:distance-AEI}`
-- {2} `\faCheckCircle{}` 
\end{lstlisting}
\end{figure}
}

\paragraph{\textbf{Incorrect strategy in computing boundary of the MIXED geometry in \GEOS{}}}
 % #14
Listing~\ref{lst:boundary} demonstrates a logic bug
caused by undefined behavior when evaluating \emph{MIXED geometry} boundaries.
According to the definition of \sqlinline{ST_Within}, 
a geometry is within another only if their interiors share at least one point. 
Both interiors of \sqlinline{g1} and \sqlinline{g2} contain \sqlinline{POINT(0 0)}, thus, the expected result is that \sqlinline{g1} is within \sqlinline{g2}.
However, \PostGIS{} incorrectly judges the relation and returns false.
This was detected by canonicalization. Replacing g2 with \sqlinline{GEOMETRYCOLLECTION(LINESTRING(0 0,1 0),POINT(0 0))} yields an inconsistent result.
A PostGIS developer in \PostGIS{} confirmed it as an upstream bug and reported it to the \GEOS{} community.
The \GEOS{} developers' feedback indicated this bug was caused by a \emph{"last-one-wins"} strategy, which dictates that a point serves as the boundary for a \sqlinline{GEOMETRYCOLLECTION} if it is the boundary of the last geometry.
In this case, \sqlinline{POINT(0 0)}, the shared interior, is considered part of the boundary of \sqlinline{g2} rather than its interior.
Developers considered adopting boundary-priority semantics to address it. 

It is a logic bug that can be detected by differential testing, as \PostGIS{} and \MySQLGIS{} output different results. 
However, when comparing the behavior of \PostGIS{} and \DuckDBSpatial{}, the bug  would be missed, as they both return the same incorrect results.

\begin{figure}[t]
\lstset{style=sqlstyle}
\begin{lstlisting}[caption={\PostGIS{} misjudges the relationship between \sqlinline{g1} and \sqlinline{g2}.}, label={lst:boundary}]
SELECT ST_Within(g1,g2) FROM (SELECT 'POINT(0 0)'::geometry As g1, 'GEOMETRYCOLLECTION(POINT(0 0),LINESTRING(0 0,1 0))'::geometry As g2);
-- {f} `\faBug{}` {t} `\faCheckCircle{}`
\end{lstlisting}
\end{figure}

\paragraph{\textbf{A logic bug in prepared geometry of \GEOS{}}}
Listing~\ref{lst:prepared} demonstrates a logic bug where the ordered id pair \sqlinline{(3,2)} is missed in the results.
\sqlinline{ST_Contains(g1, g2)} is defined to return true when \sqlinline{g1} contains \sqlinline{g2} if, and only if, all points of \sqlinline{g2} lie inside and the interiors of \sqlinline{g1} and \sqlinline{g2} share at least one point.
Since the points in the second geometry lie within the third geometry, 
and \sqlinline{POINT(3 1)} is within the interior of the third geometry, 
the third geometry should contain the second geometry.
It is noticed that the first and second geometry are the same.
However, \PostGIS{} reports the pair of (3,1), 
instead of (3,1;3,2). 
A \PostGIS{} developer identified it as an upstream bug from \GEOS{}.
Feedback from \GEOS{} developers indicated that the issue was within \emph{"prepared geometry"}, 
a component designed to speed up various topological relationship functions.
The bug was fixed within one day of confirmation.
One of the developers believes this inconsistency is widespread in \PostGIS{}, stating, \emph{“as a general proposition, every prepared variant should return the same as the non-prepared variant; I imagine there's a lot of possible issues hiding in that proposition.”}

It is a logic bug that can be detected by differential testing, as both \MySQLGIS{} and \DuckDBSpatial{} produce the correct results.

\begin{figure}[t]
\lstset{style=sqlstyle}
\begin{lstlisting}[caption={\PostGIS{} misses one of the pairs.}, label={lst:prepared}]
CREATE table t (id int, geom geometry);
INSERT INTO t (id, geom) VALUES 
    (1,'GEOMETRYCOLLECTION(MULTIPOINT((0 0),(3 1)))'::geometry),
    (2,'GEOMETRYCOLLECTION(MULTIPOINT((0 0),(3 1)))'::geometry),
    (3,'MULTIPOLYGON(((0 0,5 0,0 5,0 0)))'::geometry);
SELECT a1.id, a2.id FROM t As a1, t As a2 WHERE ST_Contains(a1.geom, a2.geom);
-- {1,1; 1,2; 2,1; 2,2; 3,1; 3,3} `\faBug{}` 
-- {1,1; 1,2; 2,1; 2,2; 3,1; `\bugPattern{3,2}`; 3,3} `\faCheckCircle{}`
\end{lstlisting}
\end{figure}

\paragraph{\textbf{A logic bug in the \PostGIS{} engine.}} Listing~\ref{lst:gist} shows a fixed logic bug related to a GIST index. 
The correct result should be 1, but \PostGIS{} incorrectly returns 0.  
We detected 7 bugs in the \PostGIS{} engine, highlighting the diversity of detected bugs. 

\ignorelst{}{
\begin{figure}[t]
\lstset{style=sqlstyle}
\begin{lstlisting}[caption={\PostGIS{} engine misses 2 pairs.}, label={lst:gist}]
CREATE TABLE t AS SELECT 1 AS id, 'POINT EMPTY'::geometry AS geom;
CREATE INDEX idx ON t USING GIST (geom);
SET enable_seqscan = false;
SELECT COUNT(*) FROM t WHERE geom ~= 'POINT EMPTY'::geometry;
-- {0} `\faBug{}` {1} `\faCheckCircle{}`
\end{lstlisting}
\end{figure}
}

\paragraph{\textbf{A bug in the RANGE functionality of \PostGIS{}}} Listing~\ref{lst:dfullywithin} shows a bug in the functionality of \sqlinline{ST_DFullyWithin} in \PostGIS{}. 
The function returns true if the geometries are entirely within the specified distance of one another.
The geometry a1 is fully within each point of a2 in the distance of 100, thus, it should return true.
A \PostGIS{} developer  pointed out, \textit{“It has become clear to me that our definition of \sqlinline{ST_DFullyWithin} is probably ‘wrong’, that is, it’s not what people think they are getting when they call it.”} 

\ignorelst{}{
\begin{figure}[t]
\lstset{style=sqlstyle}
\begin{lstlisting}[caption={A RANGE functionality fails in \PostGIS{}.}, label={lst:dfullywithin}]
SELECT ST_DFullyWithin('LINESTRING(0 0,0 1,1 0,0 0)'::geometry,'POLYGON((0 0,0 1,1 0,0 0))'::geometry,100);
-- {f} `\faBug{}` {t} `\faCheckCircle{}`
\end{lstlisting}
\end{figure}
}

\paragraph{\textbf{Patterns of inducing cases.}}
We also identified common bug-inducing patterns, classifying them under \emph{EMPTY} and \emph{MIXED geometry} categories, providing a taxonomy of bugs based on trigger case patterns.
Among all \DTLfetch{type}{type}{Logic}{Sum} logic bugs, \EMPTYPATTERN{} can be triggered by test cases containing \emph{EMPTY} elements or geometries.
These types of bugs are typically fixed within one day because pinpointing the root cause (\emph{i.e.}, empty processor) is usually easier than with other patterns.
Besides, \MIXEDPATTERN{} bugs are related to the \emph{MIXED geometry}, 
caused by various factors. 
Excluding 4 bugs whose causes were not detailed in the feedback of \MySQLGIS{} developers, 
we summarize the causes as follows.
4 logic bugs stem from errors in processing \emph{EMPTY} elements. 
Another 3 logic bugs are due to dimension processors incorrectly identifying the dimensions of \emph{MIXED geometry}.
Additionally, we detected 2 logic bugs in the \emph{"prepared geometry"}, a component for optimization in \PostGIS{}. 
Moreover, the root cause of Listing~\ref{lst:boundary} is associated with the boundary processor of \emph{MIXED geometry}.

\subsection{\textbf{Comparison to the State of the Art}}
To the best of our knowledge, \ToolName{} is the first general testing approach and tool that aims to detect logic bugs for \SDBMS{}s.
Despite this, we undertook additional manual analyses of the logic bugs detected by \APPROACHAKA{}, to ascertain if they could have been identified by previous methodologies. 
For this purpose, we employed differential testing approaches alongside TLP~\cite{10.1145/3428279}, a state-of-the-art methodology for testing relational \DBMS{}s. 
In comparing different \SDBMS{}s, we chose \PostGIS{} and \DuckDBSpatial{} to illustrate a comparison between similar systems, and \PostGIS{} with \MySQLGIS{} to demonstrate a comparison between more distinct systems.
Aside from comparing different \SDBMS{}s, we also conducted differential testing by toggling an index on and off, indicated as an approach \emph{Index}.
All \DTLfetch{type}{type}{Logic}{Sum} confirmed logic bugs detected by \APPROACHAKA{} were manually analyzed to determine their detectability by previous methodologies.
When comparing different \SDBMS{}s, we assessed the potential for bug detection by examining two factors in each report: 
(1) the functions or data in the bug-inducing case are applicable for the compared \SDBMS{}s;
(2) the bug-inducing case produced different results in the compared \SDBMS{}s.
Regarding \emph{Index}, we analyzed the impact of index presence on each bug-inducing case within a specific \SDBMS{}.
For TLP, we decomposed each bug-inducing query into three partitioning queries and examined the results; unexpected outcomes signified TLP's potential for bug detection. 

Table~\ref{tab:oracles} shows the results of logic bug detection by different methods.
Of the \DTLfetch{type}{type}{Logic}{Sum} confirmed or fixed  logic bugs, \OVERLOOKED{} were overlooked by all other methods. 
4 logic bugs could be detected by comparing \PostGIS{} and \MySQLGIS{}; however, such differential testing suffers from false alarms, due to differing function definitions.
All logic bugs were missed when comparing \PostGIS{} and \DuckDBSpatial{}, 
which is expected, given their similarity.
Two index-related bugs could be found, both theoretically detectable by the \emph{Index} method.
However, applying the \emph{Index} method heavily depends on the test case design, specifically designed to make frequent use of the index, thereby enabling more efficient application of this method. 
TLP detected one index-related bug, but missed the other, as expected, since TLP is designed for relational \DBMS{} and lacks awareness of spatial relationships.

{\small
\begin{table}[t]
\centering
\captionsetup{font={small}}
\caption{Logic bugs detection comparison.
\emph{P. vs. M.} compares the results of \PostGIS{} and \MySQLGIS{}; \emph{P. vs. D.} compares the outcomes of \PostGIS{} and \DuckDBSpatial{}. \emph{Index} compares the outcomes of \SDBMS{}s with and without index. \emph{TLP} is a state-of-the-art methodology for testing relational \DBMS{}s~\cite{10.1145/3428279}.
}
\begin{tabular}{lrrrrr}
    \toprule
      & \textbf{\APPROACHAKA{}} &  
      \textbf{\emph{P. vs. M.}} & 
      \textbf{\emph{P. vs. D.}} &
      \textbf{\emph{Index}} & \textbf{\emph{TLP}}\\
    \midrule  
    \GEOS     &  8 & 3 & 1 & 0 & 0 \\
    \PostGIS  &  8 & 0 & 0 & 1 & 1 \\ 
    \MySQLGIS &  4 & 1 & 0 & 1 & 0 \\      
    \midrule  
    Sum       & 20 & 4 & 1 & 2 & 1\\
    \bottomrule
\end{tabular}

\label{tab:oracles}
\end{table}
}

\subsection{Efficiency of \ToolName{}}

\paragraph{\textbf{Run time distribution}\label{sec:time-distribution}}
We evaluated whether the bottleneck of \ToolName{} lies in the execution time of the \SDBMS{} or in our test case generation. We varied a set of parameters, $N$, which controls the number of geometries in each run, to assist users in selecting an appropriate quantity of geometries for use in \ToolName{}.

\paragraph{Methodology.}
We varied geometry generation per run by an argument, setting $N$---the number of geometries per group---at 1, 10, 50, and 100 as configuration parameters.  
We configured \ToolName{} to generate 100 random queries in each run.
Each experiment was repeated 10 times to accommodate potential performance deviations.
Subsequently, we recorded the execution time of statements in the targeted \SDBMS{} and the total time of \ToolName{} including the run time of the \SDBMS{}.

\paragraph{Results.}
Figure~\ref{fig:efficiency_study} shows the average time spent in each configuration.
Note that the execution time of \ToolName{} includes also the time spent within the SDBMS.
The proportion of statement execution time exceeds those of \ToolName{} by 90\% when $N$ is greater than 10 in \PostGIS{} and \MySQLGIS{}.
In all configurations for \DuckDBSpatial{}, the execution time of statements accounted for more than 90\%.
Besides, in all the \SDBMS{}s,
as $N$ increases, the total runtime of \ToolName{} also increases. 
This trend is particularly notable in \DuckDBSpatial{} and \MySQLGIS{}.
When $N$ is 100, the runtime is 20 times longer than when $N$ is 10.
The results indicate that the execution time spent in the \SDBMS{} dominates the overall execution costs.
This observation has been previously noted in testing approaches for relational \DBMS{}s~\cite{10.5555/3488766.3488804}, but it had not been systematically evaluated before.
Moreover, the number of geometries impacts the performance of \ToolName{} indirectly, by increasing the execution time within the \SDBMS{}, aiding users in choosing an appropriate quantity of geometries when testing.

\pgfplotstableread[row sep=\\,col sep=&]{
	datasets & PostGIS & MySQL GIS & DuckDB Spatial \\
	2 & 47    & 89  & 165  \\
	4 & 235 & 251  & 292 \\
	6 & 586 & 1675  & 1922 \\
	8 & 1553   & 4984 & 5803 \\
}\TSPATTER

\pgfplotstableread[row sep=\\,col sep=&]{
	datasets & PostGIS & MySQL GIS & DuckDB Spatial \\
	2 & 40    & 67  & 158  \\
	4 & 225 & 224  & 283 \\
	6 & 572 & 1619  & 1911 \\
	8 & 1532   & 4889 & 5789 \\
}\TDBMS

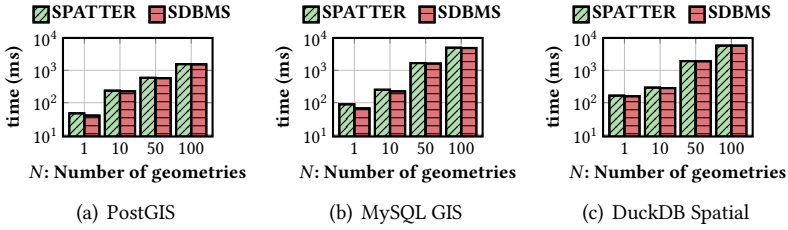
\begin{figure}[t!]
    \centering
    \subfigure[PostGIS]{
        \begin{tikzpicture}[scale=0.5]
            \tikzstyle{pattern1} = [black, fill=p1, postaction={pattern=north east lines}]
            \tikzstyle{pattern2} = [black, fill=p4, postaction={pattern=horizontal lines}]
            \begin{axis}[
                grid = major,
                ybar=0.11pt,
                bar width=0.4cm,
                width=0.4\textwidth,
                height=0.3\textwidth,
                xlabel={\LARGE \bf $N$: Number of geometries}, 
                xtick=data,	xticklabels={1, 10, 50, 100},
                 legend style={at={(0.5,1.40)}, anchor=north,legend columns=-1,draw=none},
                       legend image code/.code={
                \draw [#1] (0cm,-0.263cm) rectangle (0.4cm,0.15cm); },
                xmin=0.8,xmax=9.2,
                    ymin=9, ymax = 10000,
                     ytick = {10, 100, 1000, 10000},
            yticklabels = {$10^{1}$, $10^{2}$, $10^3$, $10^4$},
                ymode = log,    
                    log basis y={2},
                    log origin=infty,
                tick align=inside,
                ticklabel style={font=\LARGE},
                every axis plot/.append style={line width = 1.6pt},
                every axis/.append style={line width = 1.6pt},
                    ylabel={\textbf{\LARGE time (ms)}},
                ]
                \addplot[pattern1] table[x=datasets,y=PostGIS]{\TSPATTER};
                \addplot[pattern2] table[x=datasets,y=PostGIS]{\TDBMS};
            \legend{\LARGE {\tt \bf SPATTER $\ $},\LARGE {\tt \bf SDBMS $\ $}}
                \end{axis}
        \end{tikzpicture}
    }
    \subfigure[MySQL GIS]{
        \begin{tikzpicture}[scale=0.5]
            \tikzstyle{pattern1} = [black, fill=p1, postaction={pattern=north east lines}]
            \tikzstyle{pattern2} = [black, fill=p4, postaction={pattern=horizontal lines}]
            \begin{axis}[
                grid = major,
                ybar=0.11pt,
                bar width=0.4cm,
                width=0.4\textwidth,
                height=0.3\textwidth,
                xlabel={\LARGE \bf $N$: Number of geometries}, 
                xtick=data,	xticklabels={1, 10, 50, 100},
                 legend style={at={(0.5,1.40)}, anchor=north,legend columns=-1,draw=none},
                       legend image code/.code={
                \draw [#1] (0cm,-0.263cm) rectangle (0.4cm,0.15cm); },
                xmin=0.8,xmax=9.2,
                    ymin=9, ymax = 10000,
                     ytick = {10, 100, 1000, 10000},
            yticklabels = {$10^{1}$, $10^{2}$, $10^3$, $10^4$},
                ymode = log,    
                    log basis y={2},
                    log origin=infty,
                tick align=inside,
                ticklabel style={font=\LARGE},
                every axis plot/.append style={line width = 1.6pt},
                every axis/.append style={line width = 1.6pt},
                    ylabel={\textbf{\LARGE time (ms)}},
                ]
                \addplot[pattern1] table[x=datasets,y=MySQL GIS]{\TSPATTER};
                \addplot[pattern2] table[x=datasets,y=MySQL GIS]{\TDBMS};
            \legend{\LARGE {\tt \bf SPATTER $\ $},\LARGE {\tt \bf SDBMS $\ $}}
                \end{axis}
        \end{tikzpicture}
    }
    \subfigure[DuckDB Spatial]{
        \begin{tikzpicture}[scale=0.5]
            \tikzstyle{pattern1} = [black, fill=p1, postaction={pattern=north east lines}]
            \tikzstyle{pattern2} = [black, fill=p4, postaction={pattern=horizontal lines}]
            \begin{axis}[
                grid = major,
                ybar=0.11pt,
                bar width=0.4cm,
                width=0.4\textwidth,
                height=0.3\textwidth,
                xlabel={\LARGE \bf $N$: Number of geometries}, 
                xtick=data,	xticklabels={1, 10, 50, 100},
                 legend style={at={(0.5,1.40)}, anchor=north,legend columns=-1,draw=none},
                       legend image code/.code={
                \draw [#1] (0cm,-0.263cm) rectangle (0.4cm,0.15cm); },
                xmin=0.8,xmax=9.2,
                    ymin=9, ymax = 10000,
                     ytick = {10, 100, 1000, 10000},
            yticklabels = {$10^{1}$, $10^{2}$, $10^3$, $10^4$},
                ymode = log,    
                    log basis y={2},
                    log origin=infty,
                tick align=inside,
                ticklabel style={font=\LARGE},
                every axis plot/.append style={line width = 1.6pt},
                every axis/.append style={line width = 1.6pt},
                    ylabel={\textbf{\LARGE time (ms)}},
                ]
                \addplot[pattern1] table[x=datasets,y=DuckDB Spatial]{\TSPATTER};
                \addplot[pattern2] table[x=datasets,y=DuckDB Spatial]{\TDBMS};
            \legend{\LARGE {\tt \bf SPATTER $\ $},\LARGE {\tt \bf SDBMS $\ $}}
                \end{axis}
        \end{tikzpicture}
    }
    \captionsetup{font={small}}
    \caption{Average time in \ToolName{} and the \SDBMS{}s across 10 runs. \label{fig:efficiency_study} }
\end{figure}

\paragraph{\textbf{Code coverage}}
We sought to evaluate the code coverage of \ToolName{}, as it provides insights into the components of the \SDBMS{}s are tested.
\paragraph{Methodology.} 
We ran \ToolName{} on \PostGIS{} for over 24 hours, collecting line coverage data for both \PostGIS{} and \GEOS{}. 
As our testing target, we selected \PostGIS{} 3.4.3 and \GEOS{} 3.12, the versions we used when we began our testing campaign.
The reason for selecting \PostGIS{} as a representative is that we found most bugs in \PostGIS{} and its third-party library \GEOS{}.
We collected coverage data after executing its unit tests, as this allowed us to determine whether \ToolName{} had an additional  contribution to coverage on top of unit tests.
In the configuration of \emph{"Unit Tests + \ToolName{}"}, we ran \ToolName{} for 24 hours after executing all \PostGIS{} unit tests.
It is noticed that the unit tests are from \PostGIS{}, not \GEOS{}, since \GEOS{}'s unit tests cover functions (\emph{e.g.}, distance computation) that \PostGIS{} does not call.
During our testing campaign, \PostGIS{} and \GEOS{} developers incorporated our reported cases into their regression unit tests upon bug resolution.
Therefore, we selected a specific commit before our testing campaign.

\paragraph{Consequences.}
Table~\ref{tab:coverage} presents the coverage results in various settings. 
A line coverage of less than 20\% in the \PostGIS{} module might seem low, but this outcome was expected. 
The primary reason for low coverage is that 
spatial data engines encompass more components than just geometry relation processing. 
For example, in the \PostGIS{} module, 
around 15\% of the code pertains to I/O-related functions supporting various formats 
(\emph{e.g.}, GeoJSON, a JSON format describing geometry or geography), 
and 8\% of uncovered code is related to geography.
In the \GEOS{} module, the code coverage of around 20\% was also expected.
Approximately 18\% of the code, which is related to geometry operations (around 20\% of \GEOS{}), remained uncovered, because we deliberately exercised caution in calling geometry operations within the \SMARTGENERATOR{} to prevent precision issues. 
\ToolName{} is designed to detect bugs in queries concerning topological relationships rather than spatial analysis or operation functions, resulting in low function coverage.
However, coverage increased after running Spatter based on unit tests. 
Specifically, in \PostGIS{} and \GEOS{}, 
an additional 206 and 178 lines of code, respectively, were covered, indicating that \ToolName{} provides supplementary coverage beyond that of unit tests. 
These results are expected, as most logic bugs tend to occur within a limited number of code lines~\cite{10.1145/3428261}.

{\small
\begin{table}
\centering
\captionsetup{font={small}}
\caption{Code coverage of the systems tested over 24 hours.}
\begin{tabular}{lrrrr}
    \toprule
    \multirow{2}{*}{Approach} & \multicolumn{2}{c}{\PostGIS{}} & \multicolumn{2}{c}{\GEOS{}}  \\
    & Line & Function & Line & Function \\
    \midrule
    \ToolName{}                 & 15.8\% & 13.9\% & 20.1\% & 18.8\% \\
    Unit Tests                  & 79.5\% & 76.7\% & 54.8\% & 56.2\% \\
    Unit Tests + \ToolName{}    & 79.9\% & 76.8\% & 55.2\% & 56.7\% \\
    \bottomrule
\end{tabular}

\label{tab:coverage}
\end{table}
}

\paragraph{\textbf{Self-constructed baseline}}
We constructed our own baseline to evaluate the efficiency of the \genbasedstrgy{}, which, to the best of our knowledge, represents the first automated test case generation tool for \SDBMS{}s. Our baseline generates geometry based solely on the \genbasedstrgy{}, providing a comparison point to assess the efficiency of the \SMARTGENERATOR{}.

\paragraph{Methodology.}
We selected \PostGIS{} version 3.4 as our evaluation target.
The reason for selecting \PostGIS{} was that it had the highest number of bug fixes, 
allowing us to automatically determine the number of unique bugs we found during a testing campaign.
To identify unique bugs over time, we performed the following steps.
First, we ran \ToolName{} on specific older versions for one hour to collect bug-inducing cases and their trigger timestamps.
Second, for each bug-inducing test case,
we determined whether the bug was fixed by updating \PostGIS{} and \GEOS{} to their latest versions, respectively, and checking the consistency of the test case results.
If a test case was identified as fixed, we concurrently identified the module from which the bug originated.
Then, we conducted a binary search for the fix commits in the commits of \PostGIS{} or \GEOS{}. 
Finally, based on the commit and generation time of each fixed case, we determined the earliest detection time for each fixed bug during the experiment. This approach is a best-effort technique; for instance, a single fix commit might address multiple bugs. We configured each test to run for an hour. Meanwhile, we collected coverage data for \PostGIS{} and \GEOS{} over time.

\paragraph{Results.}
Our testing resulted in 2,366 and 9,913 cases potentially triggering bugs under configurations with and without the \derivativestrgy{}, respectively. 
As expected, most of them were duplicated.
We applied the deduplication technique described in our methodology to address this issue. 
As Figure~\ref{fig:ablation} shows, the generator with our proposed strategy, the \derivativestrgy{}, significantly outperformed the generator only using the \genbasedstrgy{}.
Within one hour, the \SMARTGENERATOR{} detected a higher number of unique bugs, since the \derivativestrgy{}, which generates more complex topological relationships, rendered the query results more significant.
Furthermore, the \SMARTGENERATOR{} achieved higher coverage in both \PostGIS{} and \GEOS{}, which aligns with our expectations, given that the \derivativestrgy{} exploits spatial functions inherent in \PostGIS{}.
The findings underscore that the \derivativestrgy{}, which creates new geometries based on the existing geometries, enhances 
testing efficiency.

\pgfplotstableread[row sep=\\,col sep=&]{
 BugN &   TS \\
	0 &    0 \\
	0 &    1 \\
	1 &    1 \\
	1 &   77 \\
	2 &   87 \\
	2 &   85 \\
	3 &   95 \\
	3 &  236 \\
	4 &  246 \\
	4 &  497 \\
	5 &  507 \\
    5 &  609 \\
    6 &  619 \\
    6 & 1522 \\
    7 & 1532 \\
    7 & 3600 \\
}\BugsSmart

\pgfplotstableread[row sep=\\,col sep=&]{
 BugN &   TR  \\
	0 &    0 \\
	0 &   10 \\
	1 &   15 \\
	1 &   90 \\
	2 &   95 \\
	2 &  540 \\
	3 &  547 \\
	3 & 3600 \\
}\BugsRandom

\pgfplotstableread[row sep=\\,col sep=&]{
% time & PostGIS SmartGen/RandomGen & GEOS SmartGen/RandomGen
	t &   PS &   PR &   GS &   GR  \\
	0 & 15.1 & 11.4 & 17.2 & 14.0  \\
	% 1 & 15.3 & 12.3 & 18.5 & 15.3 \\
	2 & 15.4 & 12.4 & 18.7 & 15.4 \\
	% 3 & 15.4 & 12.7 & 18.8 & 15.5 \\
	4 & 15.4 & 12.7 & 18.8 & 15.5 \\
	% 5 & 15.5 & 12.8 & 18.8 & 15.6 \\
	6 & 15.5 & 12.8 & 18.9 & 15.6 \\
}\TCoverage

\begin{figure}
    \centering
    \subfigure[Unique Bugs]{
		\begin{tikzpicture}[scale=0.8]
            \centering
            \begin{axis}[
                legend style={at={(0.5,1.35)}, anchor=north,legend columns=-1,draw=none},
				width=0.42\textwidth,
                height=0.26\textwidth,
				% legend to name,
                    % yticklabels = {$10^2$, $10^3$, $10^4$, INF},
                xtick = {0,1200,2400,3600},
                xticklabels = {0,20,40,60},
                ymax=8,
				ymin=-1,
				ytick={1,3,5,7},
				mark size=3.0pt, 
				ylabel={\normalsize \bf Bugs (\#)},
				xlabel={\normalsize \bf time (mins)}, 
                ylabel style={yshift=0pt},
                xlabel style={yshift=0pt},
				ticklabel style={font=\normalsize},
				every axis plot/.append style={line width = 1.5pt
                % , dashed, dash pattern=on 4pt off 1pt
                },
				every axis/.append style={line width = 1.5pt},
            ]
            
            % Add the first plot
            \addplot [mark=+,color=c8] table[x=TS,y=BugN]{\BugsSmart};
			\addplot [mark=x,color=c2] table[x=TR,y=BugN]{\BugsRandom};
            \legend{\large {\tt \SMARTGENERATORAKA{}}, \large{ \tt \RANDOMGENERATORAKA{}}};
            \end{axis}
        \end{tikzpicture}
	}
    \subfigure[PostGIS]{
        \begin{tikzpicture}[scale=0.5]
            \begin{axis}[
                legend style={at={(0.5,1.30)}, anchor=north,legend columns=-1,draw=none},
                width=0.4\textwidth, height=0.35\textwidth,
                % legend to name,
                    % yticklabels = {$10^2$, $10^3$, $10^4$, INF},
                xtick = {0,2,4,6},
                xticklabels = {0,20,40,60},
                ymax=16.5, ymin=9.5,
                ytick={10, 12, 14, 16},
                mark size=4.0pt, 
                ylabel={\huge \bf coverage (\%)},
                xlabel={\huge \bf time (mins)}, 
                ticklabel style={font=\huge},
                every axis plot/.append style={line width = 2.5pt},
                every axis/.append style={line width = 2.5pt},
            ]
            
            % Add the first plot
            \addplot [mark=o,color=c8] table[x=t,y=PS]{\TCoverage};
            \addplot [mark=square,color=c2] table[x=t,y=PR]{\TCoverage};
            \legend{\huge {\tt \SMARTGENERATORAKA{}}, \huge{ \tt \RANDOMGENERATORAKA{}}};
            \end{axis}
        \end{tikzpicture}
    }
    \subfigure[GEOS]{
        \begin{tikzpicture}[scale=0.5]
            \begin{axis}[
                legend style={at={(0.5,1.30)}, anchor=north,legend columns=-1,draw=none},
                width=0.4\textwidth, height=0.35\textwidth,
                % legend to name,
                    % yticklabels = {$10^2$, $10^3$, $10^4$, INF},
                xtick = {0,2,4,6},
                xticklabels = {0,20,40,60},
                ymax=20, ymin=12,
                ytick={13, 15, 17, 19},
                mark size=4.0pt, 
                ylabel={\huge \bf coverage (\%)},
                xlabel={\huge \bf time (mins)}, 
                % ticklabel style={font=\Huge},
                ticklabel style={font=\huge},
                every axis plot/.append style={line width = 2.5pt},
                every axis/.append style={line width = 2.5pt},
            ]
            
            % Add the first plot
            \addplot [mark=o,color=c8] table[x=t,y=GS]{\TCoverage};
            \addplot [mark=square,color=c2] table[x=t,y=GR]{\TCoverage};
            \legend{\huge {\tt \SMARTGENERATORAKA{}}, \huge{ \tt \RANDOMGENERATORAKA{}}};
            \end{axis}
        \end{tikzpicture}
   }
    \captionsetup{font={small}}
    \caption{Ablation study of the \SMARTGENERATORUPPER{} (\SMARTGENERATORAKA{}). \RANDOMGENERATORAKA{} denotes a generator that employs only the \genbasedstrgy{}. Sub-figure (a) shows the number of unique bugs in \PostGIS{} identified under different configurations over one hour. Sub-figure (b) and (c) illustrate the corresponding line coverage of \PostGIS{} and \GEOS{} over time. \label{fig:smart-ablation} }
    \label{fig:ablation}
\end{figure}
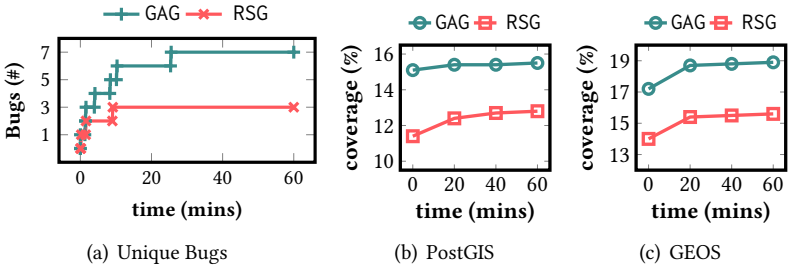

\section{Related Work}

\paragraph{Testing DBMSs}
Recently, various techniques have been proposed to detect bugs in \DBMS{}s. 
For relational \DBMS{}s, researchers have proposed methods to identify logic bugs~\cite{10.1145/3551349.3556924, 10.14778/3357377.3357382, 10.1145/3428279, 10.5555/3488766.3488804, 10.1109/ICSE48619.2023.00174, 10.1109/ICSE48619.2023.00175}, transactional bugs~\cite{10.14778/3430915.3430918, 10.1145/3360591, 10.14778/3583140.3583145, 10.1145/3514221.3517851, 288568, 10.1145/3551349.3556924}, and performance issues~\cite{10.14778/3357377.3357382, cert, 10.1145/3510003.3510093}.
A notable example is SQLancer, which integrates several novel approaches to identify a range of bug types. 
TLP~\cite{10.1145/3428279}, starting from a given original query, derives multiple, more complex queries, each of which computes a partition of the result. CERT utilizes the property of cardinality estimation to efficiently detect performance bugs.
For  graph \DBMS{}s, 
established methodologies such as \emph{differential testing}~\cite{10.1145/3533767.3534409} and TLP~\cite{10.1145/3597926.3598044} have been utilized to uncover logic bugs.
Besides, methods leveraging graph properties have been proposed to specifically address logic bugs in graph query processing~\cite{jiang2023detecting, 10.14778/3636218.3636236, 10.1145/3597503.3639200}. 
However, these testing techniques largely depend on the properties of the specific type of \DBMS{}  under test and are not directly aimed at evaluating the core functionality of \SDBMS{}s, particularly geometry processing. 
% Adapting them to test SDBMS is limited to evaluating the shared functionalities common across all types of DBMS, such as handling boolean expressions, but fails to uncover bugs highly related to geometry processing.
To the best of our knowledge, \ToolName{} represents the first approach specifically designed for \SDBMS{}s, efficiently testing both the functionalities shared with other \DBMS{} types and those unique to \SDBMS{}s.

\sloppy{}
\paragraph{\SDBMS{}s}
Various approaches have focused on studying and improving the performance of \SDBMS{}s.
A multitude of spatial indexes have been proposed, including the R-tree~\cite{guttman1984r}, R*-tree~\cite{10.1145/93597.98741}, and machine learning-based indexes~\cite{ 10.1145/3318464.3389703, 10.1145/3588917}. In addition, various optimizations for spatial joins~\cite{10.1145/2463676.2463700, 10.1145/3588716, 10.1145/2723372.2749434, 10.1145/3347146.3359343} have been proposed. 
Other advancements include learned components within \SDBMS{}s~\cite{Pandey2020TheCF}, as well as end-to-end systems~\cite{10.1145/2882903.2915237}.
For a systematic study of modern \SDBMS{}s, we defer interested readers to a study by Pandey et al. on modern \SDBMS{}s and their performance on real-world datasets~\cite{10.14778/3236187.3236213}.

\section{Discussion}

\paragraph{The reason why \APPROACHAKA{} detected bugs} We analyzed why \APPROACHAKA{} detects bugs by examining their root causes. 
Given that we check for a high-level property, \APPROACHAKA{} can detect a diverse range of bugs. 
Commonly, the reason is that the original input exercises a different path in the program as compared to the transformed follow-up database. 
For example, the bug in Listing~\ref{lst:precision-issue} is located in a function that computes the direction of a point $p$ relative to a vector $v$.
There was a loss of precision in the normalization of vertices (\emph{i.e.}, displacement to the origin). 
The statements in Listing~\ref{lst:AEI_example} fail to trigger this bug, because a point in $v$ is at the origin, leading to no displacement calculations. 
Listing~\ref{lst:distance} is another example that illustrates how \emph{canonicalization} detected bugs. 
It was caused by incorrect logic when recursively processing MULTI geometry.
Line~\ref{line:distance-AEI} in Listing~\ref{lst:distance} fails to trigger this bug, because the recursive logic is not needed to handle the second geometry.
Overall, \APPROACHAKA{} can effectively detect bugs, because the inputs before and after transformation exercise different paths.

\paragraph{Supported \SDBMS-specific functionality} 
We support various complex kinds of \SDBMS-specific functionality. For JOIN functionality, all of our test cases use JOIN to combine spatial data, given that spatial join is one of the most important functionalities in \SDBMS{}s. 
For the RANGE functionality, \SDBMS{}s support functions like \sqlinline{ST_Within}, \sqlinline{ST_DWithin}, and \sqlinline{ST_DFullyWithin} for spatial range queries. 
We have tailored functionalities based on \APPROACHAKA{} to test these functions and have indeed found significant bugs (\emph{e.g.}, see Listing~\ref{lst:dfullywithin}). 
In general, \APPROACHAKA{} can effectively test functionalities highly related to the geometry in \SDBMS{}s, such as topological relationship queries.

\paragraph{Applicability to other kinds of database systems.} \APPROACHAKA{} is a general approach and might also be applicable to other kinds of database systems. 
For example, KNN algorithms allow nearest-neighbour searching and are commonly supported not only by geospatial database systems, but also by other systems, such as vector database systems. 
It is used for classification, dimensionality reduction, and content recommendation. 
While currently not supported in our implementation, testing for KNN algorithms using \APPROACHAKA{} could be implemented as long as no shearing is applied to the geometric objects by performing the following steps: 
(1) creating a database SDB1 by the \SMARTGENERATOR{};
(2) canonicalizing each geometry in SDB1 to its equivalent representation and then applying an affine transformation (rotate, translate, or scale) to construct a new geometry, resulting in SDB2; 
(3) checking that the KNN results retrieved from SDB1 and SDB2 should be equal, otherwise, a logic bug is detected, since rotating, translating, and scaling preserve relative distance. 
Shearing cannot be applied, as it does not preserve the relative-distance property, thus the results of KNN queries before and after shearing could be inconsistent. 

\paragraph{Limitations of \APPROACHAKA{}}
While \APPROACHAKA{} is effective in identifying certain types of logic bugs, it does not cover all possible scenarios. 
First, \APPROACHAKA{} is built on affine transformations that preserve geometric properties, which do not directly apply to geography types, since they represent curved objects. 
Second, \APPROACHAKA{} is not designed to comprehensively test components tangential to query processing, such as reading and conversion of files (\emph{e.g.}, implemented via the GDAL library). 
For instance, we detected a bug in GDAL by differential testing of the \SDBMS{}s instead of \APPROACHAKA{}. \DuckDBSpatial{} was expected to accept the GeoJSON \{"type": "Polygon","coordinates": []\} as a representation of \sqlinline{POLYGON EMPTY}, but, unexpectedly, the result was \sqlinline{NULL}. 
Furthermore, \APPROACHAKA{} is inherently unsuitable for applications involving non-linear transformations or non-affine geometric operations.
For example, \sqlinline{ST_Buffer} cannot be used after \emph{affine transformations} in \APPROACHAKA{} construction, as it fails to preserve topological relationships.

\paragraph{Future work} 
There are several future directions for this work. 
First, \APPROACHAKA{} could be applied to test other geometry libraries.
We found that 12 out of 21 confirmed bugs in \PostGIS{} were in the geometry library. 
Given that the underlying geometry libraries are prone to bugs, a future direction could be to apply \APPROACHAKA{} directly to these libraries. 
It would also help to further validate the robustness and effectiveness of \ToolName{} in a different geospatial ecosystem.
Besides, symbolic execution could enhance \APPROACHAKA{} by exploring more execution paths. 
Specifically, symbolic execution can generate test cases that explore different execution paths, with \APPROACHAKA{} providing expected results, thereby replacing the \SMARTGENERATOR{}.
Furthermore, the core idea of \APPROACHAKA{} that transforms a database while preserving essential properties can be generalized to testing other DBMSs, such as GDBMSs. 
For example, after shuffling the identities between all vertices in a graph database, the results of pattern matching should be still unchanged.

\section{Conclusion}
Logic bugs in \SDBMS{}s related to spatial patterns are a severe problem. In this work, we proposed \ToolName{}, which is based on a new methodology called \APPROACHNAME{} (\APPROACHAKA{}) and a \SMARTGENERATOR{} to effectively and efficiently detect logic bugs in \SDBMS{}s. Our evaluation of \ToolName{} on four widely-used \SDBMS{}s found \DTLfetch{status}{status}{Identified}{Sum} previously unknown, and unique bugs.
\DTLfetch{status}{status}{Real}{Sum} of them have been confirmed, 
and \DTLfetch{status}{status}{Sum}{Fixed} have already been fixed. 
Further bug analysis shows that 
\OVERLOOKED{} bugs are overlooked by the previous methodologies.
Additionally, the \derivativestrgy{} in the generating process improves testing efficiency by finding more unique bugs and achieving higher coverage. 
We believe that \ToolName{}, given its simplicity and generality, 
has a high chance of being integrated into the toolbox of \SDBMS{} developers.

\begin{acks}
This research was supported by a Ministry of Education (MOE) Academic Research Fund (AcRF) Tier 1 grant, and by an Amazon Research Award Fall 2023. Any opinions, findings, and conclusions or recommendations expressed in this material are those of the author(s) and do not reflect the views of Amazon.
This research was partly conducted while Wenjing Deng and Qiuyang Mang visited the National University of Singapore.
\end{acks}

\balance
\bibliographystyle{acm}
\bibliography{main}
\end{document}